\documentclass[12pt]{amsart}

\usepackage{amsmath}
\usepackage{amssymb}

\usepackage[utf8]{inputenc}
\usepackage[T1]{fontenc}

\def\dju{\mbox{Đurđevich}}

\newtheorem{theorem}{Theorem}[section]

\newtheorem{prop}{Proposition}[section]
\newtheorem{definition}{Definition}[section]

\numberwithin{equation}{section}

\begin{document}

\begin{center}
{\bf Co-Toeplitz Operators and
\\
their Associated Quantization }
\vskip 0.3cm
\noindent
Stephen Bruce Sontz
\\
Centro de Investigaci\'on en Matematicas, A.C. 
(CIMAT)
\\
Guanajuato, Mexico 
\\
email: sontz@cimat.mx 
\end{center}

\centerline{\bf Abstract}

\vskip 0.4cm 
\noindent 
We define {\em co-Toeplitz operators}, 
a new class of Hilbert space 
operators, in order to define 
a co-Toeplitz quantization 
scheme that is dual 
to the Toeplitz quantization scheme 
introduced by the author in the setting of symbols that 
come from a possibly non-commutative algebra 
with unit. 
In the present dual setting 
the symbols come from a possibly non-co-commutative 
co-algebra with co-unit.  
However, this co-Toeplitz quantization is a usual 
quantization scheme in the sense that to 
each symbol we assign a densely defined linear operator 
acting in a fixed Hilbert space. 
Creation and annihilation operators 
are also introduced 
as certain types of co-Toeplitz operators, and then 
their commutation relations provide the way for 
introducing Planck's constant into this theory. 
The domain of the co-Toeplitz 
quantization is then extended as well 
to a set 
of {\em co-symbols}, which are the linear functionals 
defined on the co-algebra. 
A detailed example based on the quantum group 
(and hence co-algebra) 
$SU_q(2)$ as symbol space is presented. 

\vskip 0.4cm \noindent
\textbf{Keywords:} co-Toeplitz operator, 
co-Toeplitz quantization, 
creation and annihilation operators, 
second quantization.

\section{Introduction}
\label{introduction-section}

In a series of recent papers the author has introduced 
a theory of Toeplitz operators having symbols 
in a not necessarily commutative algebra with 
a {\em $*$-operation} (also called a 
{\em  conjugation}). 
See~\cite{sbs4} for the general theory and 
\cite{sbs1}, \cite{sbs2} and \cite{sbs3} for 
various examples of that theory. 
The associated Toeplitz quantization 
is also described in those papers. 
See~\cite{BC} for Toeplitz operators in 
Segal-Bargmann analysis, which was 
my original interest in these topics. 
Also see~\cite{mirek} for a quite recent review of 
Berezin-Toeplitz operators and some related topics, 
including Toeplitz operators.  
Finally, see~\cite{BandS} 
for a more general viewpoint of Toeplitz 
operators in analysis, including 
Banach space applications. 

There are at least three aspects of 
the theory in \cite{sbs4} that make it 
relevant to quantum physics. 
First, the Toeplitz operators are densely defined 
linear operators, all acting in 
the same Hilbert space, and so the 
self-adjoint extensions of the symmetric 
Toeplitz operators can be interpreted as 
being physical observables. 
(A simple sufficient condition is given in order
for a Toeplitz operator to be symmetric). 
Second, there are creation and 
annihilation operators that  
are defined as certain 
types of Toeplitz operators. 
Third, the non-zero commutation relations among the 
creation and annihilation operators allow the introduction of 
Planck's constant $\hbar$ into the theory. 

In this paper we introduce co-Toeplitz operators in order to study the 
associated dual quantization scheme. 
This opens up a new area in the well established 
theory of operators acting in Hilbert space as 
well as providing a way to quantize new types of 
`symbols' in a co-algebra.  
The most fundamental (and dual) property of the co-Toeplitz operators is that their 
symbols lie in a co-algebra rather than in an algebra as is the case for Toeplitz operators. 
A related space of `co-symbols' 
and its quantization are 
introduced as well. 
This co-Toeplitz quantization is also 
relevant to quantum physics, since it 
has the same three 
aspects as already mentioned in the 
Toeplitz setting. 

Since the co-algebra can be non-co-commutative, 
the co-Toeplitz quantization is a generalized 
{\em second quantization}, that is, it produces 
linear operators from symbols coming from an 
algebraic structure that can lack 
the appropriate commutativity, 
which for historical reasons in the case 
of co-algebras is called {\em co-commutativity}. 
In this regard it is worthwhile to note that  
P.~Dirac was famously known for saying that 
the essential property of quantum theory 
is that the observables do not commute. 
So the lack of the appropriate commutativity of 
a co-algebra makes it into a quantum object 
which the co-Toeplitz quantization then 
quantizes. 
In this sense we do have a type of second quantization. 

Some words are in order to explain 
the meaning of a {\em quantization} or a 
{\em quantization scheme}. 
I use these two expressions interchangeably. 
And I do not wish to propose a rigorous mathematical 
definition. 
The basic idea is captured in the catch-phrase
``operators instead of functions''. 
By ``operators'' I mean linear, densely defined
operators acting in a Hilbert space, 
possibly separable.
This is a quite conventional interpretation. 
But by ``functions'' I merely mean elements in 
some vector space with some additional 
algebraic structure, such as an algebra or 
a co-algebra. 
This is a far cry from the standard definition 
of a function, though that is included as 
a special case.  
The properties of the quantization mapping 
that sends ``functions'' to operators are 
left deliberately vague.  

Due to the novelty of the material of this paper, 
much of it is devoted to definitions and their 
motivation, while the number of theorems is less
than a paper of this size would usually contain. 
Some possibilities are presented in the Concluding
Remarks for research leading to more 
theorems. 
However, even the definitions may well be 
changed and refined as more examples of 
co-Toeplitz operators  become available. 

The paper is organized as follows. 
In Section~\ref{toeplitz-quantization-section} 
we review the known, general Toeplitz quantization 
scheme for algebras. 
In Section~\ref{co-toeplitz-quantization-section}
we present the 
dual co-Toeplitz quantization scheme. 
We discuss the role of the co-unit of the 
co-algebra in co-Toeplitz quantization in 
Section~\ref{co-symbols-section} and then show how 
that motivates an extension of this quantization 
scheme using {\em co-symbols} in the dual 
of the co-algebra. 
The duality between Toeplitz  
and co-Toeplitz 
operators is not as symmetric as 
one might have expected. 
This is presented in Section~\ref{duality-section}. 
Adjoints of the co-Toeplitz operators are studied 
in Section~\ref{adjoint-section}. 
Next the creation and annnihilation operators 
are defined in terms of co-Toeplitz operators in 
Section~\ref{ann-creation-section}, and then 
the canonical commutation relations among these 
operators are defined in Section~\ref{ccr-section} 
in algebraic terms. 
At this point Planck's constant $\hbar$ 
is introduced into the theory as well as the 
associated semi-classical 
algebras, for which $ \hbar > 0 $, 
and the classical algebra, 
for which $ \hbar = 0 $. 
We continue in Section~\ref{example-section} 
with an example of this new quantization scheme based  
on the quantum group (and hence co-algebra) 
$SU_q(2)$ as symbol space. 
A Toeplitz quantization of $SU_q(2)$ 
has already been presented in \cite{sbs5} 
using instead its structure 
as an algebra, but with the same sub-algebra of 
`holomorphic' elements. 
Finally, we conclude 
in Section~\ref{concluding-remarks-section} 
with remarks about possible further 
developments and alternatives of this theory. 

We only consider vector spaces over the field of complex numbers.
We use the standard notations $\mathbb{N}$ for the non-negative integers,
$\mathbb{Z}$ for all the integers, $\mathbb{R}$ for the real numbers and 
$\mathbb{C}$ for the complex numbers. 
For $ \alpha \in \mathbb{C} $ we let 
$ \alpha^{*} $ denote its complex conjugate.

\section{The Toeplitz quantization}
\label{toeplitz-quantization-section}

We will introduce the definition of a co-Toeplitz quantization 
using the Toeplitz quantization as a guide and motivation. 
Hence, we start with a review in this section 
of the already known theory of Toeplitz quantization in 
the setting of possibly non-commuting symbols 
as is developed by the author in \cite{sbs4}. 

We let $\mathcal{A}$ be an associative algebra with identity element $1 \equiv 1_{\mathcal{A}}$. 
This algebra could have a non-commutative multiplication; it will be the
symbol space for the Toeplitz quantization. 
Suppose that $\langle \cdot , \cdot \rangle_\mathcal{A}$ is a sesquilinear, complex symmetric form on 
$\mathcal{A}$; 
this form could possibly be degenerate. 
Our convention throughout is that all sesquilinear forms are anti-linear 
in the first entry and linear in the second. 
Moreover, suppose that there exists a sub-algebra 
$\mathcal{P}$ (not necessarily containing $1$) 
of $\mathcal{A}$
such that the sesquilinear form is positive definite when restricted to $\mathcal{P}$.
Then $\mathcal{P}$ is a pre-Hilbert space.
(This is one way of motivating the choice of the letter  $\mathcal{P}$
for this object. 
Another could be that $\mathcal{P}$ is a space 
whose elements are like holomorphic polynomials.) 
We let $\mathcal{H}$ denote a Hilbert space completion 
of $\mathcal{P}$ such that 
$\mathcal{P}$ is a dense subspace of $\mathcal{H}$. 
If we think of $\mathcal{P}$ as corresponding to a 
space of holomorphic polynomials, then 
$\mathcal{H}$ could be considered as a sort of 
generalization of the Segal-Bargmann space 
of holomorphic functions. 
See \cite{bargmann}. 

We let $\iota : \mathcal{P} \to \mathcal{A}$ denote 
the inclusion map, 
which is an algebra morphism. 
We suppose that there exists a projection map 
$P : \mathcal{A} \to \mathcal{P}$, that is,
$P \, \iota = id_\mathcal{P}$. 
While $P$ is assumed to be linear, it is {\em not} 
assumed to be an algebra morphism. 
In this abstract formalism the projection $P$
is rather arbitrary. 
However, one specific choice for it in 
several examples 
is given for $\phi \in \mathcal{A}$ by 
\begin{equation}
\label{specific-P}
    P \phi = 
 \sum_{j \in J} \langle \psi_j , \phi\rangle_{\mathcal{A}} \, \psi_j 
\end{equation}
where $\{ \psi_j ~|~ j \in J \}$ is an orthonormal 
set in $\mathcal{P} $ that is an orthonormal basis 
of $\mathcal{H}$. 
Of course, it must be shown that the possibly infinite 
sum on the right side of \eqref{specific-P} 
converges to an element in $ \mathcal{P} $. 
(This is trivially true if only finitely many 
of the summands are non-zero.) 
But be aware that $P$ defined 
this way is not necessarily 
an orthogonal projection, since the form 
$\langle \cdot , \cdot \rangle_{\mathcal{A}}$ 
need not be positive definite and, in fact, 
is degenerate in some examples. 

The operator $P$ could also be realized more generally 
as an extension to $\mathcal{A}$ 
of a reproducing kernel 
function that represents the identity map on 
the pre-Hilbert space $\mathcal{P}$. 
This is what is happening in \eqref{specific-P}
since the right side restricted to $\mathcal{P}$
is a reproducing function for $\mathcal{P}$. 
Since the algebra $\mathcal{P}$ can be 
non-commutative, the reproducing kernel 
need not be a function in the usual sense 
of that word and so will not have all 
(although some) of the properties of 
a reproducing kernel function. 
See \cite{sbs1} for an example of this 
more general type of reproducing kernel. 

We assume that there is a left action of  
$\mathcal{P}$ on $\mathcal{A}$, namely a linear map 
$$
\alpha :  \mathcal{P} \otimes \mathcal{A} \to \mathcal{A}
$$
satisfying the standard properties, namely 
$1 \cdot a = a $ if 
$1 = 1_{\mathcal{A}} \in \mathcal{P} $, and
$ p_1 \cdot ( p_2 \cdot a ) = ( p_1 p_2 )\cdot  a $
where $p \cdot a := \alpha ( p \otimes a )$ for 
$ p, p_1, p_2 \in \mathcal{P}$ and 
$a \in \mathcal{A} $. 
Here the juxtaposition $p_1 p_2$ means the 
multiplication of elements in $ \mathcal{P} $.  
Next, 
in anticipation of the definition of a left co-action
in Section~\ref{co-toeplitz-quantization-section}, 
we re-write this is terms of the map $\alpha$ as
$$
     \alpha (1 \otimes a) = a 
     \qquad \mathrm{and} \qquad
     \alpha ( p_1 \otimes \alpha (p_2 \otimes a) ) =
     \alpha (p_1 p_2 \otimes a) 
$$
for all $a \in \mathcal{A} $ and all 
$p_1, p_2 \in \mathcal{P} $. 
The first condition is only required 
if $ 1 \in \mathcal{P} $. 

For example, we could take $\alpha$ equal to  $\mu_\mathcal{A}$ restricted to
$\mathcal{P} \otimes \mathcal{A}$, where 
$\mu_\mathcal{A} : \mathcal{A} \otimes \mathcal{A} \to \mathcal{A}$ is 
the multiplication map of $\mathcal{A}$. 
In short, we could take 
$\alpha = \mu_{\mathcal{A}} \, (\iota \otimes id)$. 
This particular choice for $\alpha$ is the 
only place in this theory of Toeplitz operators
where we use the multiplication of $\mathcal{A}$. 
We should emphasize however that this particular 
choice for $\alpha$ closely corresponds 
to what is used in the 
classical theory of Toeplitz operators 
acting in function spaces. 

Nonetheless, 
other choices for $\alpha$, which do not use 
the multiplicative structure of $\mathcal{A}$,
are also possible. 
In such a case we can drop the assumption 
that $\mathcal{A}$ is an algebra and instead 
only assume that it is a vector space. 
However, we still want to have a $*$-structure
on $\mathcal{A}$ in order to be able to define 
creation and annihilation operators 
in Section~\ref{ann-creation-section}. 
Also a $*$-structure appropriately compatible 
with the inner 
product on $\mathcal{P}$ gives an easy way 
to find symmetric operators which then might 
be extendable to self-adjoint operators 
representing physical observables. 
This more general  
approach is presented in \cite{sbs4}. 

Given the setting of the previous paragraph we now  define Toeplitz operators. 
\begin{definition}
Suppose that  
$g \in \mathcal{A}$ and $\phi \in \mathcal{P}$. 
We introduce the notation $\phi g := \alpha (\phi \otimes g)  \in \mathcal{A}$ and define 
$$
T_g (\phi) := P (\phi g ) = P\alpha (\phi \otimes g) \in \mathcal{P}. 
$$
Then $T_g : \mathcal{P} \to \mathcal{P}$ 
is a linear map, and 
we say that $T_g$ is the {\rm Toeplitz operator with 
symbol $g$}. 
\end{definition}

The notation $\phi g$ was introduced merely to 
emphasize the similarity with 
classical Toeplitz operators.  
Another handy notation is $ \cdot \otimes g$, 
which is the linear map
$\mathcal{P} \to \mathcal{P} \otimes \mathcal{A}$ 
defined for $g \in \mathcal{A}$ and 
$\phi \in \mathcal{P}$ by 
$$
         ( \cdot \otimes g) \, \phi := \phi \otimes g. 
$$
Here is the corresponding diagram defining $T_g$ as the 
composition of these three maps: 
\begin{equation}
\label{three-maps}
         \mathcal{P} 
         \stackrel{\cdot \otimes g}{\longrightarrow} \mathcal{P} \otimes \mathcal{A} 
         \stackrel{\alpha}{\longrightarrow} \mathcal{A}
          \stackrel{P}{\longrightarrow} \mathcal{P}. 
\end{equation}
Thus the Toeplitz operator $T_g$ is defined
for each symbol for $g \in \mathcal{A}$ as 
\begin{equation}
\label{T-given-by}
        T_g := P \, \alpha \, (\cdot \otimes g) \in \mathcal{L}(\mathcal{P}). 
\end{equation}
where $\mathcal{L} (\mathcal{P} ) := 
\{ A : \mathcal{P} \to \mathcal{P} ~|~ A  \mathrm{~is~linear} \} $. 

To bring this more closely into notational accord with 
the usual definition of a 
Toeplitz operator in classical 
analysis, for each $g \in \mathcal{A}$ we define 
$$
M_g := \alpha (\cdot \otimes g) :  \mathcal{P} \to  \mathcal{A}.
$$
We note that $M_g$ is typically {\em not} an algebra morphism, 
even though both $\mathcal{P}$ and $\mathcal{A}$ are algebras. 
Then $T_g = P \, M_g$. 
Moreover, if we take $\alpha$ to be 
the restriction of the multiplication
on $\mathcal{A}$, which as was noted above is a 
possible case, 
then $M_g$ is indeed the operation of multiplication by $g$ on the right. 
(The change to get multiplication by $g$ on the left is easy enough.) 
However, even the rather general formula $M_g = \alpha (\cdot \otimes g)$ can itself
be generalized easily. 
All that we need is any linear map $\mathcal{A} \ni g \mapsto M_g$, where
$M_g : \mathcal{P} \to  \mathcal{A}$ is linear, 
that is, we need a linear map 
$ M: \mathcal{A} \to 
\mathrm{Hom}_{ \mathrm{Vect} } 
(\mathcal{P}, \mathcal{A} ) 
$, 
where $\mathrm{Hom}_{\mathrm{Vect} } (V,W)$ 
means the vector space of 
all linear maps $ V \to W$ of the vector spaces 
$V$ and $W$. 

We are using the unconventional notation $\mathcal{L}(\mathcal{P})$ 
in order to denote the complex
vector space of {\em all} the linear maps 
$A : \mathcal{P} \to \mathcal{P}$. 
Any such map $A$ can be considered as a densely defined linear operator in the
Hilbert space $\mathcal{H}$. 
We note that $A$ may or may not be a bounded operator. 
However, note that in general there are densely defined linear operators in the
Hilbert space $\mathcal{H}$ that do not lie in $\mathcal{L}(\mathcal{P})$.  
This is so for two reasons: 
First, the domain of a densely defined operator need not be equal to $\mathcal{P}$;
second, the domain 
need not be mapped to itself under 
the action of such an operator. 

The Toeplitz quantization that has been defined  
associates to each symbol $g \in \mathcal{A}$
an operator $T_g \in \mathcal{L}(\mathcal{P})$, 
which is the Toeplitz operator with symbol $g$. 
The mapping $T: \mathcal{A} \to \mathcal{L}(\mathcal{P})$ that is given by $ T : g \mapsto T_g$
is called the 
{\em Toeplitz quantization (scheme)}. 
A question that arises naturally is whether the 
Toeplitz quantization $ T $ is injective, 
that is, if a Toeplitz operator comes from 
a unique symbol. 
For example, in a certain context 
Theorem~4.3 in \cite{sbs3} says that the sesquilinear 
form on $ \mathcal{A} $ being non-degenerate is 
a necessary and sufficient condition for $ T $ to be 
injective. 
See \cite{sbs3} for more details. 

Even though $\mathcal{A}$ and 
$\mathcal{L}(\mathcal{P})$ are algebras, 
the Toeplitz quantization $ T $ 
is not expected nor desired 
to be an algebra morphism. 
On the contrary, the deviation of $ T $ from 
being an algebra morphism is some way of 
measuring the `quantum-ness' of $ T $. 
As an example, we might have elements 
$ g, h \in \mathcal{A} $ satisfying the  
`classical' $ q $-commutation relation 
$ g h - q h g = 0 $ for $ q \in \mathbb{C} $, 
while the corresponding 
Toeplitz operators satisfy the `quantum' 
$ q $-commutation relation 
$ T_g T_h - q T_h T_g = \hbar \, I_{\mathcal{P}} $. 
In Section~\ref{ccr-section} the rigorous 
definitions of `classical' and `quantum' 
relations are given in a related context. 

The identity element $1 = 1_{\mathcal{A}}$ in 
$\mathcal{A}$ has played 
no essential role so far in this theory. 
It seems that in the examples the main property
of $1$ that arises  is $T_1 = I_{\mathcal{P}}$, the identity map. 
Nonetheless, we would like to find the dual of this property in the co-Toeplitz setting. 
To achieve this requires more details about how 
$T_1$ is defined in the Toeplitz setting. 
These details are rather trivial, but their 
duals in the co-Toeplitz setting motivate 
an important definition there, as we shall see. 

Let's first note that 
$\mathrm{Hom}_{\mathrm{Vect} } (\mathbb{C} , \mathcal{A} ) \cong 
\mathcal{A} $ in a natural way. 
Explicitly, a symbol $g \in \mathcal{A}$ corresponds 
to the linear map $l_g : \mathbb{C} \to \mathcal{A}$
given by $l_g (z) := z \, g$ 
for every $z \in \mathbb{C}$. 
And an arbitrary linear map 
$l : \mathbb{C} \to \mathcal{A}$ 
has the form $l = l_g$, where $g:= l(1)$ with 
$ 1 \in \mathbb{C} $. 
Then we have that the composition  
$$
\mathcal{P} \cong \mathcal{P} \otimes \mathbb{C} 
\stackrel{id \otimes l_g}{\longrightarrow} 
\mathcal{P} \otimes \mathcal{A}
$$
is equal to $\cdot \otimes g$. 
So we can use this to re-write \eqref{T-given-by} as 
$T_{g} = P \, \alpha \, (id \otimes l_g)$. 
By taking the case where 
$g = 1_{\mathcal{A}} = 1 \in \mathcal{A}$ 
we see that $l_1 = \eta : \mathbb{C} \to \mathcal{A}$, the unit map of the algebra $\mathcal{A}$. 
By further taking $\alpha$ 
to be the restriction of the multiplication of 
$\mathcal{A}$, that is 
$\alpha = \mu_{\mathcal{A}} \, (\iota \otimes id)$, 
we easily get 
$T_1 = I_{_{\mathcal{P}}}$. 

Various examples of this sort of Toeplitz quantization have been worked out
in some of the author's papers. 
In those examples there is some sort of 
definition of a `holomorphic element' in the
algebra $\mathcal{A}$, which then must actually be a $*$-algebra, and $\mathcal{P}$ is the sub-algebra 
(but {\em not} a sub-$*$-algebra) of holomorphic elements in $\mathcal{A}$. 
There is also a concept of `anti-holomorphic element' in $\mathcal{A}$ with its corresponding
sub-algebra, defined by 
$\overline{\mathcal{P}}:= \mathcal{P^*}$, of the anti-holomorphic elements. 
Then Toeplitz operators with symbols in $\mathcal{P}$ 
are defined to be creation operators. 
On the other hand, 
Toeplitz operators with symbols in 
$\overline{\mathcal{P}}$ are defined 
to be annihilation operators. 
This aspect of the theory, which includes commutation relations among these operators, gives 
the theory contact with ideas from the mathematical physics of quantum systems. 

It might be worthwhile to recall for the record 
what a $ * $-algebra with identity $ 1 $ is. 
First off, a {\em $ * $-operation} 
(or {\em conjugation}) on a 
vector space $ V $ is an  
{\em anti-linear} 
map $ V \to V $, denoted  
by $v \mapsto v^*$ for $ v \in V $, that is also an 
involution (that is, $ v^{**} =v $). 
Then a {\em $ * $-algebra with identity~$ 1 $} 
is an algebra $ \mathcal{A} $ with identity~$ 1 $ 
which also has a  $ * $-operation satisfying 
$ (a b)^{*} = b^{*} a^{*} $ 
for all $ a,b \in \mathcal{A} $ 
as well as $ 1^{*} = 1 $. 

The conjugation in the symbol space $\mathcal{A}$ interchanges by definition the 
holomorphic and anti-holomorphic sub-algebras, namely
$$
            \mathcal{P}^* = \overline{\mathcal{P}} \quad \mathrm{and} \quad
            (\overline{\mathcal{P}})^* = \mathcal{P}. 
$$

But the Toeplitz quantization that we have described breaks this symmetry, since the 
creation and annihilation operators have distinct 
properties in specific examples.
The origin of this has to do with the fact that the 
Toeplitz operators are acting in the
holomorphic space $\mathcal{P}$, 
even though we could have used 
the anti-holomorphic space 
$\overline{\mathcal{P}}$ instead of $\mathcal{P}$. 
All of the technical details work out if we 
use $\overline{\mathcal{P}}$.  
For example, the projection of $ \mathcal{A} $ 
onto $\overline{\mathcal{P}}$ is given by 
the linear operator $ P^{*} $, 
where $ P^{*} (f) := (  P(f^{*}) )^{*}$ 
is the standard $ * $-operation 
(but {\em not} adjoint) of an operator.  
Then the Toeplitz quantization 
(which now produces operators in 
$ \mathcal{L} ( \overline{\mathcal{P}} ) $) 
of the symbols 
in $\mathcal{P}$ give the annihilation operators, 
while on the other hand 
the Toeplitz quantization of the symbols 
in $\overline{\mathcal{P}}$ give 
the creation operators. 
However, this is still to be considered 
as a type of Toeplitz quantization. 
The new concept of co-Toeplitz quantization 
comes in the next section. 

It is important to realize that the role played 
by the sesquilinear form on $ \mathcal{A} $  
is not essential to this theory. 
However, it 
does unify three different aspects of it. 
First, it can be used to define the projection
$ P $, although that can be done without having a 
sesquilinear form. 
Second, it can be used to define the left action, 
although that can also be defined independently. 
Third, it restricts to an inner product 
on $ \mathcal{P} $. 
But one can also define that inner product 
directly. 
Given these comments, we see how the sesquilinear
form, which does appear in some examples, can be 
removed from this theory without basically 
changing it.

\section{The co-Toeplitz quantization}
\label{co-toeplitz-quantization-section}

Now we continue with the dual development of the new theory of 
co-Toeplitz quantization. 
This is achieved by reversing most of the arrows in the theory of 
Toeplitz quantization as outlined in the previous section. 
This sort of duality is well known in category 
theory and is called {\em notion duality}. 
We will consider {\em object duality} in 
Section~\ref{duality-section}. 

We let $\mathcal{C}$ be a co-associative co-algebra 
with a  co-unit 
$\varepsilon : \mathcal{C} \to \mathbb{C}$ and with 
 $\Delta : \mathcal{C} \to \mathcal{C}    \otimes \mathcal{C}$, a 
possibly non-co-commutative 
co-multiplication. 
The co-algebra $\mathcal{C}$ is the 
symbol space for the co-Toeplitz quantization. 
It is important to note that even the co-commutative case is new. 
For the definition 
and basic properties of co-algebras see \cite{KS}. 

We suppose next that $\mathcal{C}$ is equipped with a 
sesquilinear, complex symmetric form 
denoted by 
$\langle \cdot , \cdot \rangle_\mathcal{C}$. 
Let $\mathcal{P}$ be  a co-associative, co-algebra 
with co-multiplication $\Delta^\prime$, 
but not necessarily with a co-unit. 
Suppose that there also exists a 
co-algebra morphism 
$Q : \mathcal{C} \to \mathcal{P}$, 
dual to $ \iota $ in the Toeplitz setting. 
Also, we suppose that there exists 
a linear map $j  : \mathcal{P} \to \mathcal{C}$, 
dual to $ P $ in the Toeplitz setting, 
such that 
$$
Q \, j = id_\mathcal{P}. 
$$

The injection $j$ need not be a co-algebra morphism. 
We suppose that the form on $\mathcal{C}$ restricts  
down using $j$ to a 
positive definite inner product 
$\langle \cdot , \cdot \rangle_\mathcal{P}$ 
on $\mathcal{P}$, that is to say, 
$\langle f , g \rangle_\mathcal{P} = \langle j(f) , j(g) \rangle_\mathcal{C}$ 
holds for all $f,g \in \mathcal{P}$. 
Therefore, $\mathcal{P}$ is a pre-Hilbert space. 
We let $\mathcal{H}$ denote a Hilbert space completion of $\mathcal{P}$ such 
that $\mathcal{P}$ is a dense subspace of $\mathcal{H}$. 
Comparing this with the Toeplitz setting, we 
notice that the arrow of the inclusion map 
of the pre-Hilbert space $\mathcal{P}$ into 
the Hilbert space $\mathcal{H}$ has not been 
reversed in the co-Toeplitz setting. 
So, it still makes intuitive sense to think 
of $\mathcal{P}$ as a space of 
`holomorphic polynomials'
and of $\mathcal{H}$ as a type of 
generalized Segal-Bargmann space
of `holomorphic functions'. 

The projection map $Q$ in this setting is 
quite abstract, although it is required to be a 
co-algebra morphism while the projection $P$ 
in the Toeplitz setting 
was only required to be linear. 
Nonetheless a similar formula using the form 
$\langle \cdot , \cdot \rangle_\mathcal{C}$ 
can be used to define $Q$ in examples. 
We will see this in the example in 
Section \ref{example-section}. 

We also suppose that there is a 
{\em left co-action} of 
the co-algebra $\mathcal{P}$ on 
$\mathcal{C}$, namely, there exists a linear map
$$
  \beta : \mathcal{C} \to \mathcal{P} \otimes \mathcal{C}
$$
which has the usual properties dual to those of 
a left action, namely, 
$$
  ( \varepsilon^\prime \otimes  id_{ \mathcal{C} }) \, \beta 
  \cong id_{ \mathcal{C} } \qquad 
  \mathrm{and} \qquad ( id_{ \mathcal{P} } \otimes \beta) \, \beta = 
  (\Delta^\prime \otimes id_{\mathcal{C}} ) \, \beta. 
$$
Each of these properties can be expressed by 
a commutative diagram. 
However, the first property is only required when 
the co-algebra $ \mathcal{P} $ has a co-unit 
$ \varepsilon^\prime $. 
As an example the left co-action 
$\beta$ could be the composition
\begin{equation}
\label{special-beta}
       \mathcal{C} \stackrel{\Delta}{\longrightarrow} \mathcal{C} \otimes \mathcal{C} 
       \stackrel{Q \otimes id}{\longrightarrow} \mathcal{P} \otimes \mathcal{C}  
\end{equation}
as the reader can readily  verify by checking that the 
corresponding diagrams commute. 
(Hint: The co-associativity of $\Delta$ is used.)  
In this case $\beta$ is a projection of the co-multiplication of $\mathcal{C}$. 
In the dual case of Toeplitz operators we had  
a particular choice of the left action $\alpha$ 
given by 
$\alpha = \mu_{\mathcal{A}} \, ( \iota \otimes id )$.  
So this particular choice of $\beta$ in  
\eqref{special-beta} 
is dual to that choice of $\alpha$ in the Toeplitz case. 
Also, much as in the Toeplitz case, this choice of 
$\beta$ is the only place in this theory of co-Toeplitz 
operators where we use the co-multiplication of  $\mathcal{C}$. 
With other choices of $\beta$ which do not depend
on the co-multiplicative structure of $\mathcal{C}$ 
we do not need to assume that $\mathcal{C}$ is a 
co-algebra. 
Rather, we only need to assume that $\mathcal{C}$ is 
a vector space equipped with a $*$-structure. 
In the example in Section~\ref{example-section}  
we will use the 
particular choice \eqref{special-beta} 
and so that example will be a 
co-algebra. 
It remains for future research work to find non-trivial 
examples of co-Toeplitz operators in a setting where
the symbol space is not a co-algebra. 

Given the set-up of the previous paragraph, we 
now define co-Toeplitz operators.
\begin{definition}
We take $g \in \mathcal{C}$, known as a {\em symbol}, 
and then consider the composition, dual 
to diagram \eqref{three-maps}, 
of these three linear maps from right to left: 
\begin{equation}
\label{dual-three-maps}
\mathcal{P} \stackrel{\pi_g}{\longleftarrow} \mathcal{P} \otimes \mathcal{C} 
 \stackrel{\beta}{\longleftarrow} \mathcal{C}  
\stackrel{j}{\longleftarrow} \mathcal{P}, 
\end{equation}
where the family of linear maps 
$\{ \pi_g ~|~ g \in \mathcal{C} \}$, 
the dual to the family of linear maps 
$\{ \cdot \otimes g ~|~ g \in \mathcal{C} \}$, 
has yet to be defined. 
Then $C_g := \pi_g \, \beta \, j$ 
is the definition of the 
{\em (left) co-Toeplitz operator with symbol $g$}. 
\end{definition}

Clearly, $C_g : \mathcal{P} \to \mathcal{P}$ is linear or, in other words, 
$C_g \in \mathcal{L}(\mathcal{P})$. 
In particular, $C_g$ is a densely defined operator 
in the Hilbert space $\mathcal{H}$. 
By replacing $\beta$ with a right co-action we get a 
theory of {\em right co-Toeplitz operators}. 
That quite similar, analogous theory 
will not be discussed here; 
we will only concern ourselves 
with left co-Toeplitz operators.  

Next, the possibly non-linear function 
$C : \mathcal{C} \to \mathcal{L}(\mathcal{P})$ 
defined by 
$g \mapsto C_g$ is called 
the {\em co-Toeplitz quantization}.  
We note in passing that the vector space 
$\mathcal{L}(\mathcal{P})$
is an algebra under the multiplication 
given by composition of operators, 
while $\mathcal{L}(\mathcal{P})$ does not seem 
to have a natural co-algebra structure. 

As in the Toeplitz setting, it is natural to ask 
whether the co-Toeplitz quantization map $ C $ 
is injective. 
It seems reasonable to conjecture that this will 
depend on other conditions, much as we 
already remarked is the case 
in the Toeplitz setting. 

Analogously to the Toeplitz case, we can introduce some notation to help
understand better what is going on here. 
In analogy to $M_g$ we define 
$$
\tilde{M}_g := \pi_g \, \beta :  \mathcal{C} \to \mathcal{P}
$$
for $g \in \mathcal{C}$.
Then $C_g = \tilde{M}_g  \, j = \pi_g \, \beta \, j \in \mathcal{L}(\mathcal{P})$. 
Be aware that $\tilde{M}_g $ maps a co-algebra to a co-algebra, 
but $\tilde{M}_g $ is not a co-algebra morphism. 
This is dual to the Toeplitz setting where 
$M_g : \mathcal{P} \to \mathcal{A} $ is a map 
between algebras, but is not an algebra morphism. 

We still have a quite general theory (possibly too general!), since 
the family $\{ \pi_g ~|~ g \in \mathcal{C} \}$ is quite arbitrary in the above discussion. 
For example, $\pi_g$ could be independent of $g$ thereby 
giving a co-Toeplitz quantization that does not depend on the symbol. 
This is much more general than we would 
wish to consider.  
A more acceptable possibility is to define 
$\pi_g :  \mathcal{P} \otimes \mathcal{C} \to \mathcal{P}$ by 
\begin{equation}
\label{define-pi-g}
\pi_g (\phi \otimes f) := \langle g , \, f  \rangle_\mathcal{C} \, \phi
\end{equation}
for $\phi \in \mathcal{P}$ and $f, g \in \mathcal{C}$. 
To see that this formula gives a dual to the
map $ \cdot \otimes g $ (now defined in the co-Toeplitz setting), 
we consider the following calculation for
$ \psi, \phi \in \mathcal{P} $ and 
$ f, g \in \mathcal{C} $: 
\begin{align*}
\langle 
(\cdot \otimes g) \psi , \phi \otimes f
\rangle_{\mathcal{P} \otimes \mathcal{C}} &=
\langle 
\psi \otimes g , \phi \otimes f
\rangle_{\mathcal{P} \otimes \mathcal{C}}
\\
&=
\langle 
\psi , \phi 
\rangle_{\mathcal{P}} 
\,
\langle 
g , f
\rangle_{\mathcal{C}} 
\\
&=
\langle 
\psi , 
\langle 
g , f
\rangle_{\mathcal{C}} \, \phi 
\rangle_{\mathcal{P}} 
\\
&=
\langle 
\psi , 
\pi_{g} (\phi \otimes f) 
\rangle_{\mathcal{P}}. 
\end{align*}
This provides some justification for 
the formula \eqref{define-pi-g}  
for $ \pi_{g} $. 
Note that the second equality here is the standard 
definition of the sesquilinear form on 
$ \mathcal{P} \otimes \mathcal{C} $. 

Now given our convention for sesquilinear forms, 
$\pi_g$ is a linear map, but in this case 
the co-Toeplitz quantization mapping  
$C : g \mapsto C_g$ is anti-linear.  
It seems to be some sort of tradition 
in mathematical physics that a 
quantization map should be linear. 
To avoid this slight unpleasantness we could 
define $ \pi_g $ by 
\begin{equation*}
\pi_g (\phi \otimes f) = \langle g^* , \, f \rangle_\mathcal{C} \, \phi
\end{equation*}
for $\phi \in \mathcal{P}$ and $f, g \in \mathcal{C}$.  
Of course, to have this make sense we must assume 
that $\mathcal{C}$ is a $ * $-co-algebra, 
which we will do anyway later.  
But, we rather prefer to let the quantization 
mapping be anti-linear. 

Again for the record let us recall that a 
{\em $ * $-co-algebra} $ \mathcal{C} $ is a 
co-algebra with a $ * $-operation 
such that the co-multiplication map
$ \Delta : \mathcal{C} \to 
\mathcal{C} \otimes \mathcal{C} $
is a {\em $ * $-morphism}, namely, 
$ \Delta (g^{*}) = ( \Delta (g) )^{*}$.
Since the co-algebra 
$ \mathcal{C} $ has a co-unit 
$ \varepsilon : \mathcal{C} \to \mathbb{C} $, 
we also require that $ \varepsilon $ 
is a $ * $-morphism, namely,  
$ \varepsilon (g^{*}) = 
( \varepsilon ( g ) )^{*}$. 
Note that the $ * $-operation of 
$ \mathcal{C} \otimes \mathcal{C} $ is 
determined by 
$(g \otimes h)^{*} = g^{*} \otimes h^{*}$
for $ g,h \in \mathcal{C} $. 
Be aware that this is not exactly dual to the
definition of a $ * $-algebra, where the 
multiplication is required to be an
anti-$ * $-morphism. 

Given the definition \eqref{define-pi-g} 
for $\pi_g$ we can write down more explicit 
expressions for $\tilde{M}_g$ and $C_g$. 
So we take $f \in \mathcal{C}$ and 
then in Sweedler's notation for a co-action 
(see Appendix~B in \cite{QPB}) we have 
$$
\beta (f) = f^{(0)} \otimes  f^{(1)} \in \mathcal{P} \otimes \mathcal{C}.
$$ 

It follows for $ g \in \mathcal{C} $ that
$$
\tilde{M}_g (f) = 
\pi_g \, \beta (f) = 
\pi_g ( f^{(0)} \otimes  f^{(1)} ) = 
\langle g , \, f^{(1)} \rangle_\mathcal{C} \, f^{(0)}. 
$$

For $C_g$ we simply note that for 
$\phi \in \mathcal{P}$ we have that
$$
C_g (\phi) = \tilde{M}_g  \, j (\phi) = 
\langle g , \, f^{(1)} \rangle_\mathcal{C} \,  f^{(0)}, 
$$
where now $f = j (\phi)$. 
If we use the injection $j$ to identify $ \mathcal{P}$ as a subspace of $\mathcal{C}$, 
then the previous expression simplifies to 
$$
C_g (\phi) =  
\langle g , \, \phi^{(1)} \rangle_\mathcal{C} \, \phi^{(0)}. 
$$

The co-action $\beta$ is a basic operation 
in these expressions. 
However, $\beta$ is hidden inside Sweedler's notation. 
For example, as noted earlier, we can take 
$\beta = (Q \otimes id) \, \Delta_\mathcal{C} 
: \mathcal{C} \to \mathcal{P} \otimes \mathcal{C}$. 
Then for $f \in \mathcal{C}$ we have 
$$
     \beta (f) = Q ( f^{(1)}) \otimes  f^{(2)}
     \quad \mathrm{and} \quad 
     C_{g} (\phi) = 
     \langle g, f^{(2)} {\rangle}_{\mathcal{C}} 
     \, Q (f^{(1)}), 
$$
where we are using Sweedler's notation 
for the co-multiplication, that is,  
$\Delta_\mathcal{C} (f) = f^{(1)} \otimes  f^{(2)} 
\in \mathcal{C} \otimes \mathcal{C}$. 
Be aware please that this is 
{\em not} Sweedler's notation 
$f^{(0)} \otimes  f^{(1)}$ 
introduced above for the co-action $\beta$. 

To maintain contact with physics ideas we only 
consider the case when 
$\mathcal{C}$ is a $*$-co-algebra. 
But, in that case we do {\em not} require 
$\mathcal{P}$ to be a sub-$*$-co-algebra. 
Rather we think of the elements 
in $\mathcal{P}$ as being {\em holomorphic} variables, 
while those in $\mathcal{P}^*$
are {\em anti-holomorphic} variables. 
Then the creation operators are defined 
to be those of the form $C_g$ 
where $g \in \mathcal{P}^{*}$, 
while annihilation operators are those 
of the form $C_g$ 
where $g \in \mathcal{P}$. 
What relation holds between the
operators 
$(C_g)^*$, the adjoint of $ C_{g} $,  
and $C_{g^*}$ for a symbol 
$g \in \mathcal{C}$ 
is a question that we will consider later. 

A possible relation between the 
sesquilinear form and the $ * $-operation 
is given in the next definition. 
This property was already described in the 
Toeplitz setting in \cite{sbs5}, 
but it was not given its own name there. 

\begin{definition}
	If for all $ f, g \in \mathcal{C} $ the identity 
	\begin{equation}
	\label{star-symmetry}
	\langle f^{*} , g^{*} \rangle_{\mathcal{C}} 
	= \langle f , g \rangle_{\mathcal{C}}^{*} 
	\end{equation}
	holds, 
    then we say that the sesquilinear form 
    $ \langle \cdot , \cdot \rangle_{\mathcal{C}} $
    is {\rm $ * $-symmetric}.	
\end{definition}

As in the Toeplitz setting it is important to 
understand the role of the sesquilinear form 
in the co-Toeplitz setting, where it has the 
same three aspects mentioned earlier 
as in the Toeplitz setting 
plus a new aspect, which is that it appears 
in the definition \eqref{define-pi-g} of $ \pi_{g} $. 
This seems to be a more essential role 
since $ \pi_{g} $ so defined 
is dual to the map $ \cdot \otimes g$ 
in the Toeplitz setting.

\section{The co-unit and co-symbols}
\label{co-symbols-section}

So far the co-unit has not played a role 
in this theory of co-Toeplitz operators. 
To achieve this we now will dualize the 
theory from the Toeplitz setting. 
Since the co-unit 
$\varepsilon : \mathcal{C} \to \mathbb{C}$ 
is a linear map, we consider how to deal 
with an arbitrary linear map 
$\lambda : \mathcal{C} \to \mathbb{C}$ in a way that 
is dual to the linear maps 
$l : \mathbb{C} \to \mathcal{A}$ 
which appeared in the Toeplitz setting. 
The dual construction, starting with $\lambda$ 
instead of with a symbol $g \in \mathcal{C} $, 
gives us a more general type of co-Toeplitz 
operator defined as the composition from 
right to left as follows: 
\begin{equation}
\label{general-type-with-e} 
\mathcal{P} \cong \mathcal{P} \otimes \mathbb{C} 
\stackrel{id \otimes \lambda}{\longleftarrow} 
\mathcal{P} \otimes \mathcal{C} 
\stackrel{\beta}{\longleftarrow} \mathcal{C}
\stackrel{j}{\longleftarrow} \mathcal{P} 
\end{equation}

However, the linear functional 
$\lambda$ lies in 
$\mathrm{Hom}_{\mathrm{Vect}} ( \mathcal{C} ,
\mathbb{C} )$ 
which, quite unlike its dual 
$\mathrm{Hom}_{\mathrm{Vect}} 
( \mathbb{C}, \mathcal{C} )$, 
is {\em not} naturally isomorphic in general to 
$\mathcal{C} $. 
Of course, for every symbol $g \in \mathcal{C}$ each 
\begin{equation}
\label{define-e-q}
e_g := \langle g, \cdot \rangle_{\mathcal{C}} 
\end{equation}
lies in 
$\mathrm{Hom}_{\mathrm{Vect}} (\mathcal{C} ,
\mathbb{C})$.   
Moreover, if we take $\lambda = e_g$ 
in diagram  \eqref{general-type-with-e}, 
we readily see that 
$$
id \otimes e_g = \pi_g 
$$
and so we do have the co-Toeplitz operators as 
defined above as a special case of the more 
general definition 
$$
  C_{\lambda} := ( id_{\mathcal{P}} \otimes \lambda  ) \, \beta \, j
$$
for  
$\lambda \in \mathrm{Hom}_{\mathrm{Vect}} 
( \mathcal{C} ,
\mathbb{C} )$. 
Having this definition in hand, it now makes 
sense to study the co-Toeplitz operator 
$C_{\varepsilon}$, where 
$\varepsilon : \mathcal{C} \to \mathbb{C}$ 
is the co-unit of the co-algebra $\mathcal{C}$. 

In the Toeplitz setting we had that 
$T_1 = I_{\mathcal{P}}$ in the special case when the 
left action was the restriction of the 
multiplication of $\mathcal{A}$. 
So in the present co-Toeplitz setting we expect 
a similar result when the left co-action 
$\beta$ is the projection of the 
co-multiplication, that is, when we have 
$\beta = (Q \otimes id_{\mathcal{C}} ) \,  
\Delta_{\mathcal{C}} $. 
In this case
for $\phi \in \mathcal{P}$ we compute that   
\begin{align*}
C_{\varepsilon} \, \phi &= 
( id_{\mathcal{P}} \otimes \varepsilon  ) \, \beta \, j \phi 
= 
( id_{\mathcal{P}} \otimes \varepsilon  ) \, 
( Q \otimes id_{\mathcal{C}} ) \, 
\Delta_{\mathcal{C}} \, \phi 
\\
&= 
 ( Q \otimes id_{\mathbb{C}} ) 
 ( id_{\mathcal{C}} \otimes \varepsilon  ) \, 
 \, \Delta_{\mathcal{C}} \, \phi 
 = 
 ( Q \otimes id_{\mathbb{C}} ) ( \phi \otimes 1 ) 
 \\
 &\cong Q \, \phi 
 = \phi, 
\end{align*}
where in the last equality we used that 
$\phi \in \mathcal{P}$. 
Also $ 1 $ here means the identity element 
$ 1 \in \mathbb{C} $. 

This discussion, which seemed at the start to be 
a minor side issue, has given rise to a new 
definition which we now explicitly state. 
\begin{definition}
	Let 
	$\lambda \in \mathcal{C}^\prime :=
	\mathrm{Hom}_{\mathrm{Vect}} ( \mathcal{C} ,
	\mathbb{C} )$ 
	be a linear functional on 
	the co-algebra $\mathcal{C}$.  
	Then we define the 
	{\rm (generalized) co-Toeplitz operator}  
	with {\rm co-symbol $\lambda$} to be 
	the linear map 
$$
   C_{\lambda} := 
   ( id_{\mathcal{P}} \otimes \lambda ) \, \beta \, j 
   \in \mathcal{L} (\mathcal{P}). 
$$

Much as before, we define the 
{\rm (generalized) co-Toeplitz quantization} 
to be the map 
$ C:\mathcal{C}^\prime \to \mathcal{L} (\mathcal{P})$
given by $\lambda \mapsto C_\lambda$  
for $\lambda \in \mathcal{C}^\prime$. 
\end{definition}

We sometimes omit the word `generalized' 
when speaking of these new objects, since 
the fact that we are using co-symbols in 
$\mathcal{C}^\prime$ rather than symbols
in $\mathcal{C}$ suffices to remove 
any ambiguity. 	
Note that the notation in this definition 
gives us the strange looking identity 
$$ 
   C_{g} = C_{e_{g}} 
$$
for any $g \in \mathcal{C}$, where 
on the left side there is a co-Toeplitz 
operator with symbol $ g $ and on the right 
side there is a generalized co-Toeplitz operator 
with co-symbol $ e_{g} $ as defined in 
\eqref{define-e-q}. 

So, given this definition,  
we have proved above the following result, which is 
dual to the result that $T_1 = I_{\mathcal{P}}$ 
in the Toeplitz setting. 
\begin{prop}
Let the left co-action be 
\begin{equation}
\label{beta-special-form}
\beta = 
( Q \otimes id_{\mathcal{C}} ) \, \Delta_{\mathcal{C}}. 
\end{equation}
Then the  
co-Toeplitz quantization of the co-unit  
$\varepsilon$ of the co-algebra $\mathcal{C}$ is 
$$
C_{\varepsilon} = I_{\mathcal{P}},  
$$
the identity operator on $ \mathcal{P} $. 
\end{prop}

So an important point here is that the set 
of co-symbols can be strictly larger than the 
set of co-symbols of the form $e_{g}$ for 
$g \in \mathcal{C}$.  
Recall that the sesquilinear form on $ \mathcal{C}$
could be degenerate, and so the Riesz 
representation theorem need not apply here. 
What we do see however is that the  
theory of co-Toeplitz operators 
with co-symbols 
can possibly admit more operators than the 
original co-Toeplitz operator 
theory with symbols only in $\mathcal{C}$. 

A more important point is that the dual space 
$\mathcal{C}^\prime$
of a co-algebra with co-unit 
has a canonical structure as 
an algebra with unit, where the multiplication 
of the elements 
$\alpha, \beta \in \mathcal{C}^\prime$ 
is defined as the composition 
$$
    \mathcal{C} 
    \stackrel{\Delta}{\longrightarrow} 
    \mathcal{C} \otimes \mathcal{C} 
    \stackrel{\alpha \otimes \beta}{\longrightarrow} 
    \mathbb{C} \otimes \mathbb{C}
    \cong \mathbb{C}
$$
and the unit is the linear map 
$\eta : \mathbb{C} \to \mathcal{C}^\prime$
defined by $ \eta (z):= z \varepsilon $
for all $z \in \mathbb{C}$, 
where $\varepsilon \in \mathcal{C}^\prime$
is the co-unit of $\mathcal{C}$. 
So the moral of this story is that the 
generalized co-Toeplitz quantization with 
co-symbols in $ \mathcal{C}^\prime $ is a 
map from the {\em algebra} $ \mathcal{C}^\prime $
to the {\em algebra} $ \mathcal{L} (\mathcal{P}) $
of linear operators. 
Of course, we do not expect this map $C$ to be an 
{\em algebra} morphism.  
Rather, as we have remarked earlier 
in the Toeplitz setting, 
the discrepancy that $ C $ has from being an 
algebra morphism is an indication of the 
`quantum-ness' of the generalized 
co-Toeplitz quantization $ C $. 

A particular case of this multiplication 
occurs by taking $\alpha = e_g$ and 
$\beta = e_h$ for $g,h \in \mathcal{C}$. 
Then for all $ \phi \in \mathcal{C} $ we get 
\begin{align*}
  (e_g \, e_h) (\phi) &= 
  (e_g \otimes e_h) (\phi^{(1)} \otimes \phi^{(2)}) 
  \\
  &=  e_g( \phi^{(1)} ) e_h( \phi^{(2)} )
  \\
  &= \langle g, \phi^{(1)} \rangle_{\mathcal{C}} \, 
     \langle h, \phi^{(2)} \rangle_{\mathcal{C}}. 
\end{align*}
Furthermore, if $ \phi $ is a group-like element 
(that is, $ \Delta (\phi) = \phi \otimes \phi $), 
then this simplifies to 
$$
   (e_g \, e_h) (\phi) = 
   \langle g, \phi \rangle_{\mathcal{C}} \, 
   \langle h, \phi \rangle_{\mathcal{C}} =
   e_g (\phi) \, e_h (\phi). 
$$

The family of the $ e_{g} $'s plays an 
important role in this theory. 
\begin{definition}
\label{define-e}
	We define 
	$ e :\mathcal{C} \to \mathcal{C}^{\prime} $ 
	by
	$ e(g):= e_g $ for all $ g \in \mathcal{C} $. 
\end{definition}

Note that $ e $ is an anti-linear map which 
need not be injective nor surjective. 
Moreover, the range of $ e $ need not be a 
subalgebra of $ \mathcal{C}^{\prime}  $. 
However, we do have the following nice property.

\begin{theorem}
Suppose the sesquilinear form is $ * $-symmetric. 
Then the range $ \mathrm{Ran} \, e $ of $ e $ 
is closed under 
the $ * $-operation of $ \mathcal{C}^{\prime}  $. 
More specifically, $ (e_{g})^{*} = e_{g^{*}} $ holds 
for all $ g \in \mathcal{C} $, that is, 
$ e $ is a $ * $-morphism. 
\end{theorem}
\noindent 
\textbf{Remark:}
The $ * $-operation of $ \mathcal{C}^{\prime}  $ 
is defined by 
$\lambda^{*} (g) := ( \lambda (g^{*}) )^{*}$
for $ \lambda \in \mathcal{C}^{\prime} $
and $ g \in \mathcal{C} $. 
\begin{proof}
We calculate for 
$ g, h \in \mathcal{C} $ that 
$$
  (e_{g})^{*} (h) =  
  \big( e_{g} (h^{*}) \big)^{*} = 
  \langle g , h^{*} \rangle_{\mathcal{C}}^{*} = 
  \langle g^{*} , h \rangle_{\mathcal{C}} = 
  e_{g^{*}} (h), 
$$
where we used the $ * $-symmetry in the 
third equality. 
This shows the second assertion of the theorem 
from which the first assertion follows directly. 
\end{proof}

It is a quite general fact that the Toeplitz 
quantization map does not preserve multiplication, 
even though it is a map between algebras. 
The co-Toeplitz quantization map does not 
preserve co-multiplication ever, since it maps
into a vector space with no natural 
co-multiplication even though its domain is a 
co-algebra. 
But the generalized co-Toeplitz quantization 
is a map from the algebra $ \mathcal{C}^{\prime} $
to the algebra $ \mathcal{L} (\mathcal{P}) $.  
And this map sends the identity element 
$ \varepsilon$ of $\mathcal{C}^{\prime} $
to the identity element 
$I_{\mathcal{P}}$ of $\mathcal{L} (\mathcal{P}) $ 
when \eqref{beta-special-form} holds. 
But what is the relation of  
the generalized co-Toeplitz quantization map 
with the multiplication? 
The next result may come as a surprise. 
\begin{theorem}
Suppose that a co-Toeplitz quantization satisfies: 
\begin{itemize}

\item 
The left co-action $ \beta $ is given by 
\eqref{beta-special-form}.  

\item  
$ \mathcal{P} $ is a sub-co-algebra 
of $ \mathcal{C} $, that is, 
$ \Delta_{ \mathcal{P}} = \Delta_{ \mathcal{C}} \!
\upharpoonright_{\mathcal{P}}$. 

\end{itemize}

\noindent 
Then the generalized co-Toeplitz quantization map 
$ 
C : \mathcal{C}^{\prime} \to \mathcal{L} (\mathcal{P}) $ 
is an algebra morphism. 
\end{theorem}
\noindent 
\textbf{Remark:}
This tells us that under 
the given hypotheses 
the generalized co-Toeplitz quantization map 
is just too nice. 
For physical reasons 
we want to have a quantization that is not 
quite so nice. 
After all, Dirac has taught us that the 
distinguishing characteristic of quantum theory is 
that the observables do not commute. 
In this more general context Dirac's insight 
can be extended to say that the range of 
a quantization mapping 
should be less commutative that its domain. 
So, in the favorable case when $ C $ is injective, 
we do not want $ C $ to be an algebra morphism. 
Therefore, I consider this to be a No~Go theorem. 
Now the hypothesis on $ \beta $ seems reasonable, 
since it is the dual of the commonly used condition
on the left action in the Toeplitz setting. 
But the second hypothesis is dual to assuming  
in the Toeplitz setting that the projection
$ P : \mathcal{A} \to \mathcal{P} $ is an 
algebra morphism. 
And that is a condition which we do not wish 
to impose. 
Hence, the second hypothesis is something which we 
want to not hold in examples and in the future 
development of this theory. 
That hypothesis does not hold in the example 
in Section~\ref{example-section}. 
Of course, a No~Go theorem is a theorem and 
is worth knowing. 

\begin{proof}
Take $ \lambda, \mu \in  \mathcal{C}^{\prime} $ 
and $ \phi \in \mathcal{P} $. 
Throughout the proof
we use the notation 
$ \Delta :=  
\Delta_{ \mathcal{P}} = \Delta_{ \mathcal{C}}
 \! \upharpoonright_{\mathcal{P}}
$, which comes from the second hypothesis. 
We also use the iterated Sweedler notation as 
explained, for example, in \cite{QPB}. 
Then we calculate as follows. 
\begin{align*}
    C_{\lambda} \, C_{\mu} \phi &= 
    (id \otimes \lambda) (Q \otimes id) \, \Delta \,  
    (id \otimes \mu) (Q \otimes id) \, \Delta \phi  
    \\
    &= 
    (id \otimes \lambda) (Q \otimes id) \, \Delta \, 
    (id \otimes \mu) (Q \phi^{(1)} \otimes \phi^{(2)})
    \\ 
    &= 
    \mu(\phi^{(2)}) 
    (id \otimes \lambda) (Q \otimes id) \, \Delta \, 
    \phi^{(1)}   
    \\ 
    &= 
    \mu(\phi^{(2)})  
    (id \otimes \lambda) (Q \otimes id) \, 
    ( \phi^{(11)} \otimes \phi^{(12)} )  
    \\ 
    &= 
    \mu(\phi^{(2)}) \,
    \lambda(\phi^{(12)}) \, 
     \phi^{(11)}     
     \\ 
     &= 
     \mu(\phi^{(22)}) \,
     \lambda(\phi^{(21)}) \, 
     \phi^{(1)}  
     \\
     &=
     \big( 
        (\lambda \otimes \mu) \, \Delta \phi^{(2)} 
     \big)
     \phi^{(1)} 
     \\
     &=
     \big( (\lambda \mu) \phi^{(2)} \big) Q \phi^{(1)} 
     \\
     &=
     (id \otimes \lambda \mu) 
     (Q \otimes id) \, \Delta \phi 
     \\
     &=
     C_{\lambda \mu} \, \phi. 
\end{align*}

Here we used 
$\Delta \phi = \phi^{(1)} \otimes \phi^{(2)} 
\in \mathcal{P} \otimes \mathcal{P}$, 
the fact that $ Q $ acts as the identity 
on $ \mathcal{P} $, the co-associativity of 
$ \Delta $, the definition of the product 
$  \lambda \mu   $ and the definition of 
the co-Toeplitz quantization mapping $ C $. 
The first hypothesis was used in the first and 
last equalities. 
\end{proof}

The question naturally arises whether the 
generalized co-Toeplitz quantization map is 
injective. 
Using Definition~\ref{define-e} 
and Equation~\eqref{define-e-q}, 
we see that a necessary condition 
for this injectivity is that 
$ e : \mathcal{C} \to \mathcal{C}^{\prime} $ 
is injective, which itself is equivalent to 
the sesquilinear form on $ \mathcal{C} $ 
being non-degenerate. 
 
The extension of the co-Toeplitz quantization 
from the domain of symbols to the domain 
of co-symbols leads one to wonder 
if there is a corresponding extension of the 
domain of the Toeplitz quantization. 
Now the symbol $ g \in \mathcal{A} $ in the 
Toeplitz setting was used there to define a linear
map $l_{g} : \mathbb{C} \to \mathcal{A}$. 
And this map $ l_{g} $ was all that we needed 
to define the Toeplitz operator with symbol~$ g $. 
But the generalization given by replacing 
$ l_{g} $ with an arbitrary linear map 
$l : \mathbb{C} \to \mathcal{A}$ is no 
generalization at all because, 
as noted earlier, any such map $ l $ is equal 
to $ l_{g} $ for a unique symbol 
$ g \in \mathcal{A} $. 
So the co-Toeplitz quantization shows a bit 
of flexibilty, let's say, that is not present 
in the Toeplitz quantization. 
This is an indication of a lack of symmetry 
between the Toeplitz and co-Toepliz quantizations, 
a topic that we will consider in more detail 
in the next section.

\section{Duality}
\label{duality-section}

We now discuss in detail in what sense the 
theories of Toeplitz and co-Toeplitz 
quantization are duals of each other. 
The duality behind the definition 
of co-Toeplitz operators 
comes about simply by reversing the direction 
of all the arrows (i.e., morphisms)
in the definition of a Toeplitz operator. 
This sort of duality comes from category theory 
and is seen in the 
formulation of the basic 
concepts of non-commutative geometry, for 
example. 
It is called {\em notion duality}. 
This is exactly what we see in the relation between 
the definitions \eqref{three-maps} and 
\eqref{dual-three-maps} of 
Toeplitz operators 
and of co-Toeplitz operators, respectively. 

However, another sort of duality (called 
{\em object duality}) arises from 
applying the {\em duality contravariant functor} 
$ V \mapsto V^{\prime} 
\equiv \mathrm{Hom}_{\mathrm{Vect}} (V, \mathbb{C}) $ 
for $ V $ a complex vector space 
and the corresponding pull-back definition 
$T \mapsto T^{\prime} : W^{\prime} \to V^{\prime}$ for 
a morphism (i.e., linear map) $ T: V \to W $ 
of vector spaces $ V $ and $ W $.  
Specifically, 
$ T^\prime (\lambda) := 
\lambda \circ T \in V^{\prime}$ for 
$ \lambda \in W^{\prime} $.  
So the question arises as to what happens to 
\eqref{three-maps} and 
\eqref{dual-three-maps} when we apply this duality
contravariant functor to each of them. 
Of course, we do get some operator. 
The question is what type of operator it is and 
whether it has a simple formula. 

One nice property is that a $ * $-operation on $ V $ 
induces a $ * $-operation on $ V^{\prime} $ 
defined by 
$ \lambda^{*} (v) := ( \lambda (v^{*}) )^{*}  $ 
for $ \lambda \in V^{\prime} $ and $ v \in V $. 
Let us also recall from the last section that the dual 
$\mathcal{C}^{\prime}$ 
of a co-algebra $ \mathcal{C} $ is always an algebra. 
On the other hand, the dual $ \mathcal{A}^{\prime}  $ 
of an algebra $ \mathcal{A} $ is not necessarily 
a co-algebra. 
Briefly, the point is that in general 
the duality contravariant functor is only 
{\em sub-multiplicative} 
with respect to the tensor product, namely, 
$
V^{\prime} \otimes W^{\prime} \subset 
(V \otimes W)^{\prime}. 
$ 
However, if either $ V $ or $ W $ is 
finite dimensional, then the duality  
contravariant functor is {\em multiplicative}, 
$
V^{\prime} \otimes W^{\prime} = 
(V \otimes W)^{\prime}. 
$
To get multiplicativity in the full infinite 
dimensional setting requires changing either the 
definition of the duality contravariant functor or 
the definition of the tensor product (or of both). 
See \cite{KS} for more details. 
A rather similar analysis, which we leave to the 
interested reader, shows that the dual 
of a co-action is always an action, while the dual 
of an action is not necessarily a co-action. 

But the dual of a vector space with a sesquilinear 
form does not in general have a naturally defined 
sesquilinear form. 
So, we will not look for a full duality between 
Toeplitz and co-Toeplitz operators using this 
duality contravariant functor. 
Thus, we will mainly consider the duality relation 
between the diagrams \eqref{three-maps} and 
\eqref{dual-three-maps} considered as diagrams 
of {\em vector spaces} as well as the {\em definitions} 
of Toeplitz and co-Toeplitz operators, respectively. 
Similarly, we take \eqref{general-type-with-e}
to be the diagram of 
{\em vector spaces} which {\em defines} 
a generalized co-Toeplitz operator. 
But we will comment on other 
algebraic aspects of this duality contravariant 
functor as they arise in specific contexts. 

Given this situation, it seems more feasible 
for us to first consider the dual of a co-Toeplitz 
operator as defined in \eqref{dual-three-maps} 
with symbol $ g \in \mathcal{C} $, a co-algebra, 
which gives us this dual diagram: 
\begin{equation}
\label{this-dual-diagram}
\mathcal{P}^{\prime} 
\stackrel{\pi_{g}^{\prime}}{\longrightarrow} 
(\mathcal{P} \otimes \mathcal{C})^{\prime} 
\stackrel{\beta^{\prime}}{\longrightarrow} 
\mathcal{C}^{\prime} 
\stackrel{j^{\prime}}{\longrightarrow} 
\mathcal{P}^{\prime}. 
\end{equation}
To understand this diagram 
we evaluate $ \pi_{g}^{\prime} $. 
So for $ \lambda \in \mathcal{P}^{\prime} $, 
$ \phi \in \mathcal{P} $ and  
$ f,g \in \mathcal{C} $ we have  
\begin{align*}
    &\pi_{g}^{\prime} (\lambda)(\phi \otimes f) =  
    ( \lambda \circ \pi_{g} ) (\phi \otimes f) 
    = 
    \lambda 
    \big( 
    \langle g , f \rangle_{\mathcal{C}} \, \phi  
    \big) 
    \\
    &= 
     \langle g , f \rangle_{\mathcal{C}} \,
    \lambda 
    \big( 
    \phi  
    \big) 
    =
    e_{g} (f) \lambda (\phi) 
    =
    ( \lambda \otimes e_{g} ) (\phi \otimes f), 
\end{align*}
which implies that 
$ \pi_{g}^{\prime} (\lambda)
= \lambda \otimes e_{g} 
\in \mathcal{P}^{\prime} \otimes \mathcal{C}^{\prime}$ 
and hence 
$$ 
\pi_{g}^{\prime} = \cdot \otimes e_{g} : 
  \mathcal{P}^{\prime} \to 
  \mathcal{P}^{\prime} \otimes \mathcal{C}^{\prime} 
  \subset (\mathcal{P} \otimes \mathcal{C})^{\prime}. 
$$ 
Then \eqref{this-dual-diagram} becomes 
\begin{equation}
\label{better-dual-diagram}
\mathcal{P}^{\prime} 
\stackrel{ \cdot \otimes e_{g} }{\longrightarrow} 
\mathcal{P}^{\prime} \otimes \mathcal{C}^{\prime} 
\stackrel{\beta^{\prime}}{\longrightarrow} 
\mathcal{C}^{\prime} 
\stackrel{j^{\prime}}{\longrightarrow} 
\mathcal{P}^{\prime}. 
\end{equation}
This is a Toeplitz operator as 
defined by \eqref{three-maps} with 
symbol $ e_{g} $ in the algebra 
$ \mathcal{C}^{\prime} $. 
Moreover, $ \beta^{\prime} $ is a left action 
and $ j^{\prime} $ is a projection. 
Also 
$ Q^{\prime} : \mathcal{P}^{\prime} \to \mathcal{C}^{\prime}$ 
is a unital algebra morphism. 
We have the following. 

\begin{theorem}
\label{dual-of-co-toepltz-operator}
If $ C_{g} \in \mathcal{L} (\mathcal{P}) $ 
is a co-Toeplitz operator with symbol $ g $ 
in the co-algebra $ \mathcal{C} $, 
then 
$ (C_{g})^{\prime} = T_{e_{g}} \in  \mathcal{L} (\mathcal{P}^{\prime})$ 
is a Toeplitz operator with symbol $ e_{g} $ in 
the algebra $ \mathcal{C}^{\prime} $. 

If $ C_{\mu} \in \mathcal{L} (\mathcal{P}) $ 
is a generalized 
co-Toeplitz operator with co-symbol $ \mu $ 
in the algebra $ \mathcal{C}^{\prime} $, 
then 
$ (C_{\mu})^{\prime} = T_{\mu} \in  \mathcal{L} (\mathcal{P}^{\prime})$ 
is a Toeplitz operator with symbol $ \mu $ in 
the algebra $ \mathcal{C}^{\prime} $. 
\end{theorem}

\noindent 
\textbf{Remark:}
We can also write the result of the first part 
as $ (C_{e_{g}})^{\prime} = T_{e_{g}} $. 

\begin{proof}
We have already proved the first assertion above. 
As for the second assertion we note that 
in the above argument the symbol $ g $ is used to 
define the linear functional 
$ e_{g} \in \mathcal{C}^{\prime}$, 
which is the only 
occurrence of $ g $ in \eqref{better-dual-diagram}. 
So we replace $ e_{g} $  with the co-symbol 
$ \mu $ in that argument to obtain 
$ (\cdot \otimes \mu)
: \mathcal{P}^{\prime} \to 
\mathcal{P}^{\prime} \otimes \mathcal{C}^{\prime} $ 
 in \eqref{better-dual-diagram}, 
and the second result follows immediately. 
\end{proof}

On the other hand, the 
dual of a Toeplitz operator 
is not necessarily a co-Toeplitz operator. 
To see this we examine the dual of diagram 
\eqref{three-maps}, which is 
\begin{equation}
\label{yet-another-dual}
\mathcal{P}^{\prime} 
\stackrel{(\cdot \otimes g)^{\prime}}{\longleftarrow} (\mathcal{P} \otimes \mathcal{A})^{\prime} 
\stackrel{\alpha^{\prime}}{\longleftarrow} 
\mathcal{A}^{\prime}
\stackrel{P^{\prime}}{\longleftarrow} \mathcal{P}^{\prime}
\end{equation}
Here neither $ \mathcal{A}^{\prime} $ nor 
$\mathcal{P}^{\prime} $ need be a co-algebra 
although each does have a $ * $-operation. 
Consequently, it need not make sense in general to 
require $ P^{\prime} $ to be a co-algebra morphism. 
Recall that  
`co-Toeplitz operator' (resp., `generalized 
co-Toeplitz operator')
now means the composition 
of the maps of {\em vector spaces} in diagram 
\eqref{dual-three-maps} 
(resp., diagram \eqref{general-type-with-e}). 

Even if $ \mathcal{P}^{\prime} $ is a co-algebra, 
$ \alpha^{\prime} $ need not be a left 
co-action on $ \mathcal{A}^{\prime} $, since 
$$ 
\mathcal{P}^{\prime} \otimes \mathcal{A}^{\prime}
\subset (\mathcal{P} \otimes \mathcal{A})^{\prime} 
$$ 
can be a proper inclusion. 
But we do have the following result. 

\begin{theorem}
\label{dual-of-toeplitz-operator}
If $T_{g} \in \mathcal{L} (\mathcal{P}) $ 
is a Toeplitz operator with symbol $ g $ 
in the algebra $ \mathcal{A} $ and 
the left action 
$ \alpha : \mathcal{P} \otimes \mathcal{A} 
\to \mathcal{A} $ (used to define the 
Toeplitz operator) satisfies 
$ \mathrm{Ran} \, \alpha^{\prime} \subset 
\mathcal{P}^{\prime} \otimes \mathcal{A}^{\prime} $, 
then 
$$
(T_{g})^{\prime} = C_{\mathrm{ev}_{g}} \in 
\mathcal{L}  (\mathcal{P}^{\prime}) 
$$ 
is a generalized co-Toeplitz operator 
with co-symbol 
$ \mathrm{ev}_{g} \in \mathcal{A}^{\prime\prime} $. 
(We will define $ \mathrm{ev}_{g} $
in the course of the proof.)  
\end{theorem}
\begin{proof}
We take $ g \in \mathcal{A} $, 
$ \phi \in \mathcal{P} $, 
$ \lambda \in \mathcal{P}^{\prime} $
and $ \omega \in \mathcal{A}^{\prime}$. 
Then we calculate 
\begin{align*}
\big( 
(\cdot \otimes g)^{\prime} (\lambda \otimes \omega) 
\big) 
(\phi) &= 
(\lambda \otimes \omega) 
\big( 
(\cdot \otimes g) (\phi) 
\big)
= 
(\lambda \otimes \omega) (\phi \otimes g)  
\\
&=
\lambda (\phi) \omega (g) 
=
\big( \omega (g) \lambda \big) (\phi),  
\end{align*}
which implies 
$ (\cdot \otimes g)^{\prime} (\lambda \otimes \omega)
= \omega (g) \lambda 
= (id \otimes \mathrm{ev}_{g}) 
(\lambda \otimes \omega)  $, where 
$ \mathrm{ev}_{g} (\omega) := \omega (g) $ 
defines the evaluation functional 
$ \mathrm{ev}_{g} $ at $ g $. 
Let's note that 
$\mathrm{ev}_{g} \in \mathcal{A}^{\prime\prime}$
does hold. 
Therefore we have arrived at 
$$
   (\cdot \otimes g)^{\prime} = 
   id \otimes \mathrm{ev}_{g}. 
$$
So \eqref{yet-another-dual} becomes 
$$
\mathcal{P}^{\prime} 
\stackrel{ id \otimes \mathrm{ev}_{g} }{\longleftarrow} \mathcal{P}^{\prime} \otimes \mathcal{A}^{\prime} 
\stackrel{\alpha^{\prime}}{\longleftarrow} 
\mathcal{A}^{\prime}
\stackrel{P^{\prime}}{\longleftarrow} \mathcal{P}^{\prime}, 
$$
where we also used the hypothesis on the 
range of $ \alpha^{\prime} $. 
And so we have shown 
that $ (T_{g})^{\prime} $ is the generalized 
co-Toeplitz operator $ C_{\mathrm{ev}_{g}} $.  
\end{proof} 

These two theorems show an asymmetry in this 
duality, namely, the dual of a co-Toeplitz operator
is always a Toeplitz operator while for a 
Toeplitz operator we used an extra technical 
hypothesis in order to show that its dual is a 
co-Toeplitz operator. 
Of course, this opens the door to the possibility 
of altering the definition of Toeplitz operator
(and maybe of co-Toeplitz operator as well) 
in the infinite dimensional case 
in order to obtain a more precise duality. 

We are now in a position to evaluate the 
double duals of Toeplitz and co-Toeplitz operators. 
It is an elementary fact that the double dual 
always exists. 
What we want to do is describe it explicitly. 
Here is some well known material that we are  
going to use in order to study double duals. 
\begin{definition}

Suppose that V is a vector space and that $v \in V$. 
Then we define 
$\mathrm{ev}^{V}_{v} \in V^{\prime\prime}$,
the {\rm evaluation at $ v $}, by 
$$
   \mathrm{ev}^{V}_{v} (f) := f(v) 
$$
for all $ f \in V^{\prime} $. 
We also define the {\rm evaluation map} 
$$
\mathrm{ev} \equiv 
\mathrm{ev}^V : V \to V^{\prime\prime} 
$$
by 
$\mathrm{ev}^V (v) := \mathrm{ev}^V_{v}$ 
for all $ v \in V $. 
We sometimes write $\mathrm{ev}$ instead of 
$\mathrm{ev}^V$ when the context indicates  
what the vector space $ V $ is. 
\end{definition}

We state the next elementary result without proof. 
\begin{prop}
	\label{without-proof}
The map $ \mathrm{ev}^V  $ is linear and 
injective. 
For any linear map $ T: V \to W $ between 
vector spaces $V$ and $W$ we have that this 
diagram commutes: 	
\begin{equation*}
 \begin{array}{rcl}
 V & \stackrel{\mathrm{ev}^V}{\longrightarrow} &  V^{\prime\prime}
 \\
  T~ \big\downarrow &  & 
  \big\downarrow ~ T^{\prime\prime} 
 \\
  W & \stackrel{\mathrm{ev}^W}{\longrightarrow} &  W^{\prime\prime}
 \end{array}
\end{equation*}
Using $ \mathrm{ev}^V $ to identify $ V $ as 
a subspace of $ V^{\prime\prime} $ 
(and similarly for $ W $), we can read this diagram 
as saying that the restriction of 
$T^{\prime\prime}$ to the subspace $ V $ is $ T $, 
that is, 
$ T^{\prime\prime} \!\! \upharpoonright_{V} = T $. 
Equivalently, $ T^{\prime\prime}  $ can be 
viewed as an extension of $ T $. 
\end{prop}

We now proceed to the theorem about double duals. 
\begin{theorem}
	There are three cases of a double dual. 
\begin{itemize}
\item 
Let $ g \in \mathcal{C} $ be a symbol and
let $ C_{g} \in \mathcal{L}(\mathcal{P})$ 
be its associated co-Toeplitz operator. 
If the map $ \beta $ used in defining $ C_{g} $ 
satisfies 
$ \mathrm{Ran} \, \beta^{\prime\prime} \subset 
\mathcal{P}^{\prime\prime} \otimes 
\mathcal{C}^{\prime\prime}$, 
then 
$ (C_{g})^{\prime\prime} = C_{\mathrm{ev}_{e_{_{g}}}} 
\in \mathcal{L}(\mathcal{P}^{\prime\prime})$. 

\item 
Let $ \mu \in \mathcal{C}^{\prime} $ be a co-symbol 
and $ C_{\mu} \in \mathcal{L}(\mathcal{P})$ 
be its associated generalized co-Toeplitz operator. 
If $ \beta $ satisfies the condition in the previous 
part of this theorem, then 
$ (C_{\mu})^{\prime\prime} = 
C_{\mathrm{ev}_{\mu}} 
\in \mathcal{L}(\mathcal{P}^{\prime\prime})$. 

\item 
Let $ g \in \mathcal{A} $ be a symbol and
$ T_{g} \in \mathcal{L}(\mathcal{P})$ 
be its associated Toeplitz operator. 
Suppose that the left action $ \alpha $
used in the 
definition of $ T_{g} $ satisfies the technical 
condition in Theorem~\ref{dual-of-toeplitz-operator}.  
Then 
$ (T_{g})^{\prime\prime} = 
T_{\mathrm{ev}_{g}} 
\in \mathcal{L}(\mathcal{P}^{\prime\prime})$. 
\end{itemize}

\end{theorem}

\noindent 
\textbf{Remark:} 
By Proposition~\ref{without-proof} 
in each of these three cases 
the double dual of the initially 
given operator is necessarily 
an extension of that operator. 
The question is whether the double dual 
of a Toeplitz (resp., co-Toeplitz) operator 
is again a Toeplitz (resp., co-Toeplitz) 
operator and, if so, what is the formula 
for the double dual. 
This theorem answers that question provided 
a specific technical condition holds. 

\begin{proof}
By Theorem~\ref{dual-of-co-toepltz-operator} 
we have $ (C_{g})^{\prime} = T_{e_{g}} $ 
for $ g \in \mathcal{C} $. 
Taking the dual of this using 
Theorem~\ref{dual-of-toeplitz-operator} gives 
$$ (C_{g})^{\prime\prime} = (T_{e_{g}})^{\prime} 
  = C_{\mathrm{ev}_{(e_g)}}
$$
using the hypothesis on $ \beta $. 
This shows the first part of the theorem. 

For the second part we have from 
Theorem~\ref{dual-of-co-toepltz-operator} 
that $ (C_{\mu})^{\prime} = T_{\mu} $ 
for a co-symbol $ \mu$ in the algebra 
$\mathcal{C}^{\prime} $. 
Then by Theorem~\ref{dual-of-toeplitz-operator} 
we obtain 
$$
(C_{\mu})^{\prime\prime} = (T_{\mu})^{\prime} 
= C_{\mathrm{ev}_{\mu}} 
$$
where we again use the same hypothesis on $ \beta $. 

For the last part 
from Theorem~\ref{dual-of-toeplitz-operator} 
we have 
$ (T_{g})^{\prime} = C_{\mathrm{ev}_{g}} $, 
using the hypothesis on $ \alpha $.  
Then applying 
Theorem~\ref{dual-of-co-toepltz-operator} 
we immediately get 
$$
(T_{g})^{\prime\prime} = 
(C_{\mathrm{ev}_{g}})^{\prime} = T_{\mathrm{ev}_{g}}. 
$$

This concludes the proof. 
\end{proof}

\noindent 
A consequence of this section is that the dual  
of a co-Toeplitz operator is a Toeplitz operator 
and has a relatively simple formula. 
However, the corresponding result for the dual of a 
Toeplitz operator required an extra hypothesis. 
So this is an asymmetry in this duality. 
Another question is whether every Toeplitz 
(resp., co-Toeplitz) operator is the dual of a 
co-Toeplitz (resp., Toeplitz) operator. 
This question remains as an open problem.

\section{Adjoints}
\label{adjoint-section}

We next examine the relation between the 
operator adjoint 
$(C_g)^*$ of a co-Toeplitz operator $C_g$ 
with symbol $ g \in \mathcal{C} $ 
and the co-Toeplitz operator $C_{g^*}$.  
Since $ C_{g} : \mathcal{P} \to \mathcal{P} $
and the vector space $ \mathcal{P} $ 
does not in general have 
a $ * $-operation on it, there should be no confusion 
with the adjoint notation $(C_g)^*$ and the 
previously defined $ * $-operation of an 
operator that maps between vector spaces with 
a $ * $-operation. 

As one would expect, to get a result we need 
to assume some sort of a relation between the 
inner product on the pre-Hilbert space 
$ \mathcal{P} $, used to define $(C_g)^*$, 
and the $*$-operation in the symbol space, 
used to define $C_{g^*}$. 
In the Toeplitz case the relation needed 
is easily seen to be 
\begin{equation}
\label{Toeplitz-star-prod-reln} 
\langle M_{g^{*}} \phi, \psi \rangle_{\mathcal{P}} 
= 
\langle \phi , M_{g} \psi \rangle_{\mathcal{P}}
\qquad \mathrm{or} \qquad 
\langle \phi g^{*}, \psi \rangle_{\mathcal{A}} 
= 
\langle \phi , \psi g \rangle_{\mathcal{A}}  
\end{equation}
for $\phi, \psi \in \mathcal{P}$ and 
$ g \in \mathcal{A} $. 
This translates directly into 
$ T_{g^{*}} \subset (T_g)^* $, an inclusion 
of densely defined operators 
acing in $ \mathcal{H} $.  
For more details, including examples, 
see \cite{sbs5}. 

For the co-Toeplitz case with symbol 
$ g \in \mathcal{C} $, a co-algebra, 
we do two straightforward calculations 
using the formula $ C_g = \tilde{M}_{g} \, j $. 
In the following we take $\phi, \psi \in \mathcal{P}$ 
and $g \in \mathcal{C}$. 
First we have
$$
\langle \phi , C_g \, \psi \rangle_{\mathcal{P}}
= 
\langle \phi , (\tilde{M}_{g} \, j) \, \psi 
\rangle_{\mathcal{P}} 
= 
\langle \phi , \tilde{M}_{g} \, \psi 
\rangle_{\mathcal{P}}. 
$$

On the other hand we get 
$$
\langle C_{g^{*}} \phi , \psi \rangle_{\mathcal{P}}
= 
\langle (\tilde{M}_{g^{*}} \, j) \phi , \psi \rangle_{\mathcal{P}} 
= 
\langle \tilde{M}_{g^{*}} \phi , \psi \rangle_{\mathcal{P}}. 
$$

So the condition we impose now and for 
the rest of this paper is
\begin{equation}
\label{M-tilde-condition}
\langle \tilde{M}_{g^{*}} \phi , \psi \rangle_{\mathcal{P}} 
=
\langle \phi , \tilde{M}_{g} \, \psi 
\rangle_{\mathcal{P}} 
\end{equation}
for all $ \phi, \psi \in \mathcal{P} $ 
and $ g \in \mathcal{C} $. 
We have shown the next result. 

\begin{theorem}
Assume \eqref{M-tilde-condition} holds. 
Then we have this 
inclusion of operators acting in $ \mathcal{H} $: 
\begin{equation}
\label{C-g-star-both-ways}
C_{g^{*}} \subset (C_g)^*. 
\end{equation} 
In particular, the adjoint of $ C_{g} $ 
restricted to $ \mathcal{P} $ is 
exactly $ C_{g^{*}} $. 
\end{theorem}

So far the argument closely follows the 
Toeplitz case. 
Replacing $ g $ with $ g^{*} $ 
in \eqref{C-g-star-both-ways} we obtain 
$C_{g} \subset (C_{g^{*}})^*$, which implies 
by functional analysis that $ C_{g} $
is a closable operator. 
Also, for $g$ real, that is $ g^{*} = g $, 
we see directly from \eqref{C-g-star-both-ways} that 
$ C_g $ is a symmetric operator, in which case 
it then becomes relevant to analyze its self-adjoint 
extensions, if such extensions exist. 
In particular, it would be interesting to know 
if $ C_{g} $ is essentially self-adjoint. 

The condition \eqref{M-tilde-condition} 
can be expanded out in various special cases. 
We use the special case for $ \beta $ 
given in \eqref{special-beta} and the 
definition of $ \pi_g $ in \eqref{define-pi-g}. 
In the following calculations we take 
$ \phi, \psi \in \mathcal{P} $
and $ g \in \mathcal{C} $. 
So, on the one hand we have 
\begin{align}
\label{M-tilde-g-one-hand} 
\langle \phi, \tilde{M}_g \, \psi \rangle_{\mathcal{P}} 
&=  
\langle \phi, \pi_g \beta \,\psi \rangle_{\mathcal{P}}
\\
&=
\langle \phi, \pi_g (Q \otimes id) \Delta_{\mathcal{C}} \, \psi \rangle_{\mathcal{P}} \nonumber
\\
&=
\langle \phi, \pi_g (Q \psi^{(1)} \otimes \psi^{(2)} \rangle_{\mathcal{P}} \nonumber
\\
&=
\langle \phi, \langle g , \psi^{(2)} \rangle_{\mathcal{C}} \, Q \psi^{(1)} 
\rangle_{\mathcal{P}} \nonumber 
\\
&=
\langle g , \psi^{(2)} \rangle_{\mathcal{C}} \, 
\langle \phi, Q \psi^{(1)} 
\rangle_{\mathcal{P}}. \nonumber
\end{align}

On the other hand, using this result \eqref{M-tilde-g-one-hand}, 
we see that 
\begin{align*}
\langle \tilde{M}_{g^{*}} \, \phi,  \psi \rangle_{\mathcal{P}}
&= 
\langle \psi, \tilde{M}_{g^{*}} \, \phi \rangle_{\mathcal{P}}^{*}
\\
&=
\big( \langle g^{*},\phi^{(2)} \rangle_{\mathcal{C}} \, 
\langle \psi, Q \phi^{(1)} 
\rangle_{\mathcal{P}} 
\big)^{*}
\\
&=
\langle \phi^{(2)} , g^{*} \rangle_{\mathcal{C}} \, 
\langle Q \phi^{(1)} , \psi 
\rangle_{\mathcal{P}} 
\end{align*}
So we have obtained the following result. 

\begin{theorem}
With the above choices 
for $ \beta $ and $ \pi_g $ 
we get that the symmetry condition 
\eqref{M-tilde-condition} 
is equivalent to  
$$
\langle g , \psi^{(2)} \rangle_{\mathcal{C}} \, 
\langle \phi, Q \psi^{(1)} 
\rangle_{\mathcal{P}} 
=
\langle \phi^{(2)} , g^{*} \rangle_{\mathcal{C}} \, 
\langle Q \phi^{(1)} , \psi 
\rangle_{\mathcal{P}} 
$$
for all $ \phi, \psi \in \mathcal{P} $
and $ g \in \mathcal{C} $. 
\end{theorem}
The condition 
in this theorem does not seem 
to be the dual of the condition 
\eqref{Toeplitz-star-prod-reln} 
in the 
Toeplitz setting, although it actually is.

\section{Creation and Annihilation Operators} 
\label{ann-creation-section}

We now come back to one of the most important 
aspects of this theory. 
First, we give the basic definition. 
\begin{definition}
    Let $ g \in \mathcal{P}^{*} $
    (or, equivalently, $ g^{*} \in \mathcal{P} $)
    be given. 
    Then we define 
    $$ 
  A^{\dagger} (g) := C_g \in \mathcal{L}(\mathcal{P}), 
    $$ 
    the {\rm creation operator (associated to 
   	the anti-holomorphic symbol $ g $)}. 
   Let $ g \in \mathcal{P} $ be given. 
   Then we define  
   $$ 
   A ( g ) := C_{g} \in \mathcal{L}(\mathcal{P}), 
   $$ 
   the {\rm annihilation operator (associated to 
   	the holomorphic symbol $ g $)}. 
\end{definition}

\noindent
\textbf{Remark:} One way to extend this 
definition to include the 
generalized co-Toeplitz operators 
is to extend to the co-symbols the 
definitions of holomorphic and anti-holomorphic 
elements. 
We leave this topic for future consideration. 
We also bring to the reader's attention 
that in the Toeplitz setting 
the holomorphic (resp., anti-holomorphic) symbols 
give the creation (resp., annihilation) operators. 
These relations are inverted in the co-Toeplitz 
setting. 
The motivation for this reversal comes from the 
example in Section~\ref{example-section}. 

\vskip 0.2cm 
These definitions are originally 
motivated by the definitions in 
Segal-Bargmann analysis and its generalizations.  
See 
Bargmann's paper~\cite{bargmann} where creation and 
annihilation operators were realized for the 
first time as adjoints of each other, which is 
basically the case here 
when \eqref{M-tilde-condition} holds. 
In this formulation the annihilation operators 
could have been defined without a $ * $-structure, 
while the creation operators use explicitly 
the $ * $-structure. 
This is just a consequence of using 
$ \mathcal{P} $ as the pre-Hilbert space. 
If the sesquilinear form is $ * $-symmetric
(see \eqref{star-symmetry}), then 
$ \mathcal{P}^{*} $ is a pre-Hilbert space 
with inner product given by restricting 
the sesquilinear form 
$  \langle \cdot , \cdot \rangle_{\mathcal{C}} $ 
to $ \mathcal{P}^{*} $. 
This is so, since 
for all $ f, g \in \mathcal{P}^{*} $ 
the identity \eqref{star-symmetry} implies 
\begin{equation}
\label{p-star-inner-prod}
   \langle f , g \rangle_{\mathcal{P}^{*}} =
   \langle f , g \rangle_{\mathcal{C}} =
    \langle f^{*} , g^{*} \rangle_{\mathcal{C}}^{*} = 
     \langle g^{*} , f^{*} \rangle_{\mathcal{P}}, 
\end{equation}
which shows that we do get a positive definite 
inner product on $\mathcal{P}^{*}$. 
Then the completion of the pre-Hilbert space 
$\mathcal{P}^{*}$ is denoted as $\mathcal{H}^{*}$. 
We can think of these as the space of anti-holomorphic 
polynomials $ \mathcal{P}^{*} $
and the anti-holomorphic 
Segal-Bargmann space $ \mathcal{H}^{*} $. 
The identity \eqref{p-star-inner-prod}
can be re-written as 
$$
  \langle f , g \rangle_{\mathcal{P}} = 
  \langle g^{*} , f^{*} \rangle_{\mathcal{P}^{*}} 
$$
which says that the anti-linear bijective map 
$ V : \mathcal{P} \to \mathcal{P}^{*} $
given by $ V f := f^{*} $ is anti-unitary. 
Also, $ V^{-1} = V $. 
Therefore, we next define the  
co-Toeplitz operator 
$\tilde{C}_{g} \in \mathcal{L}(\mathcal{P}^{*}) $ 
for $ g \in \mathcal{C} $ by 
$ \tilde{C}_{g} := V C_{g} V^{-1} $. 
This gives us essentially the same set-up as we had above, except now 
with the co-Toeplitz operators acting in a dense subspace 
of an anti-holomorphic Hilbert space. 
In this new set-up 
an annihilation operator is defined as 
$ \tilde{C}_{g} $ for $ g \in \mathcal{P}^{*} $, 
that is, the conjugation of 
a creation operator acting in the 
holomorphic Hilbert space $ \mathcal{H} $. 
Similarly, we define a creation operator 
acting in the anti-holomorphic Hilbert space  
as $\tilde{C}_{g}$ for $g \in \mathcal{P}$, 
the conjugation by $ V $ of an annihilation operator 
acting in the holomorphic Hilbert space. 

Some related structures are defined next. 
\begin{definition}
   	
   The unital subalgebra of 
   $ \mathcal{L} (\mathcal{P}) $ 
   generated by all of the creation and annihilation 
   operators is called the {\rm 	
   canonical commutation relations (CCR) algebra} 
   and is denoted as $ \mathcal{CCR} $. 
   
   The unital subalgebra of 
   $ \mathcal{L} (\mathcal{P}) $ 
   generated by all of the co-Toeplitz 
   operators with symbols in $ \mathcal{C} $ is called the 
   {\rm co-Toeplitz algebra}. 
      
   Finally, the unital subalgebra of 
   $ \mathcal{L} (\mathcal{P}) $ 
   generated by all of the generalized co-Toeplitz 
   operators with co-symbols in 
   $ \mathcal{C}^{\prime} $ is called the 
   {\rm generalized co-Toeplitz algebra}.   
\end{definition}

Creation and annihilation operators have a multitude 
of applications in physics. 
The CCR algebra also arises in many parts of physics. 
However, the newly introduced co-Toeplitz algebra 
and the generalized co-Toeplitz algebra are objects 
that are of more interest in the area of 
operator theory in mathematics. 
While all of these algebras have their 
importance, it seems that very little can be
said about them in general. 
However, they all can be studied in specific 
examples of this theory.

\section{Canonical Commutation Relations} 
\label{ccr-section}

The algebra $ \mathcal{CCR} $ defined here 
can be studied 
in much the same way as the canonical commutation 
algebra is studied in \cite{sbs5} in the Toeplitz setting. 
The upshot is that Planck's constant $ \hbar $ 
will be introduced into the theory and semi-classical 
algebras as well as a dequantized (or classical) 
algebra will be defined. 
To make this paper more self-contained we review how 
the relevant material of \cite{sbs5} applies in 
the co-Toeplitz setting. 

Note that we have already defined the algebra 
$\mathcal{CCR}$. 
It still remains to define the canonical 
commutation relations themselves. 
In physics one usually 
defines the algebra of 
canonical commutation relations by explicitly  
using generators and their relations, 
where these relations 
are by very definition the canonical commutation relations. 
In this setting we do the opposite by 
starting with 
$\mathcal{CCR}$, then writing it 
as the quotient of a free
algebra $\mathcal{F}$ and next identifying the kernel of the quotient map $p : \mathcal{F} \to \mathcal{CCR}$
as the ideal of canonical commutation relations. 
Finally, any minimal set of generators of
this ideal serves as canonical commutation 
relations associated to $\mathcal{CCR}$. 

To achieve this we define $ \mathcal{F} $ 
to be the free unital algebra generated 
by the abstract set 
$ F = \{ G_{f} ~|~ f 
\in\mathcal{P} \cup \mathcal{P}^{*} 
\subset \mathcal{C} \}$ 
in bijective correspondence with the set 
$ \mathcal{P} \cup \mathcal{P}^{*} $. 
The unital algebra morphism 
$p : \mathcal{F} \to \mathcal{CCR}$ is then defined 
on the algebra generators $ G_{f} $ 
of $ \mathcal{F} $ by $p(G_{f}) := C_{f}$ for all 
$ f \in \mathcal{P} \cup \mathcal{P}^{*} $. 
By the universal property of the free algebra 
$ \mathcal{F} $ this uniquely defines the unital 
algebra morphism $ p $. 
And since by definition the elements $ C_{f} $ for 
$ f \in \mathcal{P} \cup \mathcal{P}^{*} $ 
generate $ \mathcal{CCR} $ as a unital algebra, we 
see that $ p $ is surjective. 

\begin{definition}
	We define the 
	{\rm ideal of the canonical commutation relations (CCR)} 
    of the co-Toeplitz quantization $ C $
	to be $ \mathcal{R} := \ker \, p $. 
	
	A {\rm set of 
	canonical commutation relations (CCR)} 
    of the co-Toeplitz quantization $ C $ is defined 
    to be any minimal subset of 
    ideal generators of the two-sided ideal 
    $ \mathcal{R} $. 
\end{definition}

Notice that not only is {\em a} set of 
canonical commutation relations not unique in general, 
even its cardinality in general will not be uniquely 
determined by the given co-Toeplitz quantization. 

The free algebra $ \mathcal{F} $ has a natural 
grading 
$ \deg (G_{f_{1} } \cdots G_{f_{n} } ) := n $
for integer $ n \ge 1 $ and 
$ f_{1}, \dots , f_{n} \in \mathcal{P} \cup \mathcal{P}^{*} $. 
We also put $ \deg(1) := 0$, 
where $ 1 \in \mathcal{F} $ is the 
identity element. 
This leads to an important definition. 

\begin{definition}
	A homogeneous
	element with respect to this grading 
	 in $ \mathcal{R} $ is called a 
	{\rm classical relation}  
	while a non-homogeneous element
	in $ \mathcal{R} $ is called a 
	{\em quantum relation}.
\end{definition}

The motivation for the previous definition is 
given in \cite{sbs5}. 
While this definition applies to any element 
in $ \mathcal{R} $, its main intent is to 
divide the elements in a set of CCR into 
two disjoint subsets. 

It turns out that a 
logically possible, though physically 
anomalous, situation happens when 
$ \mathcal{R} = \ker \, p = 0$, in which case 
$ p $ is an algebra isomorphism and the (unique!)
set of CCR's is empty. 
In this strange case the quantization 
is {\em over-quantized} 
in the sense that there are no pairs 
$f_{1} \ne f_{2} \in \mathcal{P} \cup \mathcal{P}^{*}$
with the 
{\em classical (or trivial) commutation relation} 
$ C_{f_{1}} C_{f_{2}} - C_{f_{2}} C_{f_{1}} = 0 $, 
and then, as we will see momentarily, we can not
introduce Planck's constant $\hbar$ into the theory. 
Also, despite Dirac's insistence on the importance 
of non-commuting observables, some 
non-trivial and useful 
classical commutation relations are 
always present in quantum theory. 

The next definition is also motivated in the 
discussion in \cite{sbs5}. 
\begin{definition}
	Let 
	$ R \in \mathcal{R} $ be a non-zero relation. 
	Then we write $R$ 
	uniquely as 
	\begin{equation}
	\label{R-expansion}
	    R = R_{0} + R_{1} + \cdots + R_{n}, 
	\end{equation}
	where $ R_{i} $ is homogeneous with 
	$ \deg R_i = i $ (for all $ i = 0, 1, \dots , n $ 
	which satisfy $ R_{i} \ne 0$) 
	and $ R_{n} \ne 0 $. 
	
	Then we say that $ R_{n} $ is the 
	{\em classical relation associated to $R$}.
\end{definition}

Note that $ R_{n} $ is indeed a 
non-zero classical relation. 
Based on what is true in the Toeplitz setting 
as is presented in \cite{sbs5}, 
I conjecture that both of the cases 
$ R_{n} \in \mathcal{R} $ and 
$ R_{n} \notin \mathcal{R} $ can occur. 
The intuition here is that the terms 
$ R_{0}, R_{1}, \dots , R_{n-1} $ 
are `quantum corrections' to the classical 
relation $ R_{n} $. 
To see what that means let us define 
the $\hbar$-deformation of a non-zero relation 
$ R \in \mathcal{R} $ to be 
\begin{equation}
\label{define-R-h}
   R(\hbar) := \hbar^{n/2} R_{0} + 
               \hbar^{(n-1)/2} R_{1} + 
               \cdots + 
               \hbar^{1/2} R_{n-1} + 
               R_{n}, 
\end{equation}
where $ \hbar^{1/2} \in \mathbb{C} $ is arbitrary, 
$ \hbar = (\hbar^{1/2})^{2} $ and 
$ R $ is written as in \eqref{R-expansion}. 
Notice that $ R(0) = R_{n} $. 
This says that the classical 
case $\hbar = 0$ gives us the classical relation 
associated to $ R $. 

In physics we take $ \hbar^{1/2} > 0 $, but for now 
there is no need to impose that restriction. 
We use these definitions to define some more 
two-sided ideals in $ \mathcal{F} $ and their
associated quotient algebras. 
\begin{definition}
	Let $ \mathcal{R}_{cl} $ denote the two-sided 
	ideal in $ \mathcal{F} $ generated by all 
	the classical relations with degree $ \ge 1 $. 
	
	The {\rm dequantized (or classical) algebra} 
	of the co-Toeplitz quantization is defined as: 
	$$
	    \mathcal{A}_{cl} =  
	    \mathcal{DQ}:= \mathcal{F} / \mathcal{R}_{cl}. 
	$$ 
	
	Let $\mathcal{R}_{\hbar}$ denote the two-sided 
	ideal in $ \mathcal{F} $ generated by all 
	the relations $ R(\hbar) $ as defined in 
	\eqref{define-R-h} with 
	$ 0 \ne R \in \mathcal{R} $ and 
	$\deg R \ge 1 $. 
	
	Then the {\rm $\hbar$-deformed CCR algebra} 
	associated 
	with the co-Toeplitz quantization is defined
	as: 
	$$
	\mathcal{CCR}_{\hbar} :=\mathcal{F} / \mathcal{R}_{\hbar}.
	$$
\end{definition}
 
By the above remarks we see that 
$ \mathcal{DQ} = \mathcal{CCR}_{0}$. 
Also, we have $ \mathcal{CCR} = \mathcal{CCR}_{1} $. 
There seems to be no reason why 
the dequantized (or classical) algebra
$ \mathcal{DQ} $ should be commutative, 
and so I conjecture that there are examples 
where it is not. 

The algebras $ \mathcal{CCR}_{\hbar} $ 
may have limiting properties as $ \hbar > 0 $
tends to zero. 
These would be the {\em semi-classical} properties 
of the co-Toeplitz quantization. 
And properties of the algebra $ \mathcal{DQ} $
would be the {\em classical} properties 
of the co-Toeplitz quantization. 
In short, this gives us a framework for 
analyzing semi-classical as well as classical 
aspects of this theory. 
However, it seems difficult to delve into 
all this in greater detail 
at the present abstract level, 
though these considerations can be brought to bear 
on specific examples. 
The reader can consult \cite{sbs5} 
for more details, including motivation, for the 
topics of this section. 

Let me emphasize that the approach here is the 
opposite of the usual approach in mathematical physics, 
where one takes certain interesting 
commutation relations to be the given CCR's, and then 
representations of 
those same commutation relations are realized by  
operators acting in some Hilbert space, 
often a Fock space of some sort. 
This more usual approach is found in the recent paper 
\cite{bozejko} and many of the 
papers in its list of references. 
Here, on the other hand, we start with a 
Hilbert space and then define 
the creation and annihilations operators acting in it. 
Only after this do we finally arrive 
at a definition of the CCR's.

\section{An example: $SU_q(2)$}
\label{example-section}

This general theory of co-Toeplitz quantization should 
be fleshed out with specific examples. 
We now proceed with such an example. 

We let $\mathcal{C} = SU_q (2)$ for 
$0 \ne q \in \mathbb{R}$. 
To avoid technicalities we assume as well that 
$ q \ne -1$. 
Then $ SU_q (2) $ is a Hopf $*$-algebra,  
and so in particular it is a $*$-co-algebra. 
We first review some of the well-known 
facts concerning the quantum group $SU_q (2)$. 
For these and many more details see \cite{KS}. 

$SU_q (2)$ can be defined as the universal $*$-algebra 
with the identity element $1$ generated by 
elements $a$ and $c$ satisfying these relations: 
\begin{align}
\label{SU-q-2-relations}
&a c = q \, c a, \qquad a c^* = q \, c^* a, \qquad c c^* = c^* c, 
\\
&a^* a + c^* c =1, \qquad a a^* + q^2 c^* c = 1.
\nonumber 
\end{align} 

The co-multiplication 
$\Delta_\mathcal{C} : \mathcal{C} \to \mathcal{C} \otimes \mathcal{C}$ 
of this co-algebra is the unique $*$-algebra 
morphism determined by 
\begin{align*}
\Delta_\mathcal{C} (a) &= a \otimes a - q \, c^* \otimes c, 
\\
\Delta_\mathcal{C} (c) &= c \otimes a + a^* \otimes c. 
\end{align*}
The co-unit 
$\varepsilon : \mathcal{C} \to \mathbb{C}$ 
is the unique $*$-algebra 
morphism determined by 
\begin{equation*}
\varepsilon(a) = 1 \qquad  \mathrm{and} \qquad 
 \varepsilon(c) = 0. 
\end{equation*}

Even though only the  $ * $-co-algebra structure 
of  $SU_q (2)$ will be used, for completeness 
we also note that the antipode, 
denoted by $S$, is the unique unit preserving, 
anti-multiplicative algebra 
morphism (but \textit{not} $*$-morphism) 
determined by 
$$
    S(a) = a^*, \qquad S(a^*) = a, \qquad S(c) = - q c, \qquad  S(c^*) = - q^{-1} c^*. 
$$

While $SU_q (2)$ is generated by just 
two elements as a $*$-algebra, 
it is an infinite dimensional vector space. 
A Hamel basis of $SU_q (2)$ is given by 
$\{ \varepsilon_{klm} ~|~ k \in \mathbb{Z}
\mathrm{~and~} l,m \in \mathbb{N} \}$, 
where 
\begin{align*}
     \varepsilon_{klm} &= a^k \, c^l \, (c^*)^m \qquad 
     \qquad \mathrm{if~} k \ge 0, 
     \\
     \varepsilon_{klm} &= (a^*)^{-k} \, c^l \, (c^*)^m 
     \qquad \mathrm{~if~} k < 0.
\end{align*}

We define a sesquilinear form on 
$\mathcal{C} = SU_q (2)$ by requiring 
\begin{equation}
\label{define-sequi-form}
\langle 
\varepsilon_{klm}, \varepsilon_{rst} 
\rangle_{\mathcal{C}} 
= w(k, l-m) \, \delta_{k,r} \, \delta_{l-m, s-t} 
\end{equation}
and then extending anti-linearly in the first entry 
and linearly in the second entry. 
Here 
$w : \mathbb{Z} \times \mathbb{Z} \to (0, \infty)$ is 
some strictly positive weight function, and 
$\delta_{i,j}$ is the Kronecker delta function 
for $ i,j \in \mathbb{Z} $. 
See~\cite{sbs2} for motivation for how such 
a formula is related with
the inner product defined in the holomorphic 
Hilbert space in Bargmann's paper \cite{bargmann}. 

While \cite{bargmann} was the original motivation 
for \eqref{define-sequi-form}, there is another 
way of understanding this, which we now sketch. 
See \cite{KS} for more details and 
background. 
It turns out that there is an algebraic 
direct sum decomposition 
\begin{equation}
\label{c-direct-sum}
    \mathcal{C} = 
    \oplus_{m,n} A[m,n], 
\end{equation}
where the sum is over 
$ (m,n) \in \mathbb{Z} \times \mathbb{Z} $. 
This is defined in terms of two  
co-actions on $ \mathcal{C} $
of the \textit{diagonal} quantum group 
$ \mathcal{K} = \mathbb{C}[t,t^{-1}] $, the algebra of 
Laurent polynomials in the variable $ t $. 
One realizes $ \mathcal{K} $ 
(which actually is a Hopf $ * $-algebra) 
as a quantum subgroup
of $ \mathcal{C} $ via 
the surjection $ \pi : \mathcal{C} \to \mathcal{K} $
which is defined to be the algebra morphism 
determined by $ \pi (a) = t $, $ \pi(a^{*}) = t^{-1} $
and $ \pi(c) = \pi(c^{*}) = 0 $. 
Then the left co-action $ L_{\mathcal{K}} $ of
$\mathcal{K}$ on $\mathcal{C} $
is defined as the composition 
$$
    \mathcal{C} 
    \stackrel{\Delta_{\mathcal{C}}}{\longrightarrow} 
    \mathcal{C} \otimes \mathcal{C} 
    \stackrel{\pi \otimes id}{\longrightarrow}
    \mathcal{K} \otimes \mathcal{C}.  
$$
Similarly, the right co-action $ R_{\mathcal{K}} $ of
$\mathcal{K}$ on $\mathcal{C} $
is defined as the composition 
$$
\mathcal{C} 
\stackrel{\Delta_{\mathcal{C}}}{\longrightarrow} 
\mathcal{C} \otimes \mathcal{C} 
\stackrel{id \otimes \pi}{\longrightarrow}
\mathcal{C} \otimes \mathcal{K}.  
$$

Using these co-actions we define for $ m,n \in \mathbb{Z} $ 
$$
   A[m,n] := \{ x \in \mathcal{C} ~|~ 
               L_{\mathcal{K}} (x) = t^{m} \otimes x
               \quad \mathrm{and} \quad 
               R_{\mathcal{K}} (x) = x \otimes t^{n}
             \},
$$
the vector subspace of \textit{bi-homogeneous} elements with 
respect to these co-actions. 
For such a bi-homogeneous element $ x \in A[m,n] $ 
we write 
$ \mathrm{bideg} (x) = (m,n) 
\in \mathbb{Z} \times \mathbb{Z}$, a group.  
One can show that this bi-grading is compatible 
with the multiplication in $ \mathcal{C} $ in the
sense that  
\begin{equation}
\label{compatible-product}
   A[m,n] \, A[p,q] \subset A[m+p,n+q]
\end{equation}
for $ m,n,p,q \in \mathbb{Z} $, since
$ L_{\mathcal{K}} $ and $ R_{\mathcal{K}} $ 
are algebra morphisms. 
This can alternatively be written as 
$$ 
\mathrm{bideg} (xy) = 
\mathrm{bideg} (x) + \mathrm{bideg} (y) 
$$ 
for all bi-homogeneous elements $ x $ and $ y $. 
We also have that $ a \in A[1,1 ]$
and $ c \in A[-1,1] $.
Moreover, $ x \in A[m,n] $ implies that 
$ x^{*} \in A[-m,-n] $. 
Another fact is that $ A[m,n] = 0 $
if and only if $ m - n $ is odd. 

From \eqref{compatible-product} we can see that
$ A[0,0] $ is a sub-algebra of $ \mathcal{C} $
and then that each $ A[m,n] $ 
is an $ A[0,0] $-bimodule. 
One has that $ A[0,0] = \mathbb{C}[\zeta] $, 
the polynomial algebra in the variable 
$ \zeta = q^{2} c c^{*} $. 
(The coefficient $ q^{2} $ makes this notation 
conform with that in \cite{KS}.)
Furthermore, each subspace $ A[m,n] $ with $ m - n $ even 
is a free left (respectively, right)
$ \mathbb{C}[\zeta] $-module 
on one generator denoted as $ e_{m,n} $ 
in the notation of \cite{KS}. 

The basis elements $ \varepsilon_{klm} $ of $ \mathcal{C} $
turn out to be bi-homogeneous with 
$ \mathrm{bideg} (\varepsilon_{klm}) = (k-l+m, k+l-m)$
for all $ k \in \mathbb{Z} $ and $ l,m \in \mathbb{N} $. 
Since the weight function in \eqref{define-sequi-form} is 
strictly positive we see that 
$ \langle 
\varepsilon_{klm}, \varepsilon_{rst} 
\rangle_{\mathcal{C}} 
\ne 0 $
if and only if 
both $ k = r $ and $ l-m = s-t $. 
But this last condition is equivalent to both 
$ k-l+m = r-s+t $ and $ k+l-m = r+s-t $, which is 
the same as 
$ \mathrm{bideg} (\varepsilon_{klm}) = \mathrm{bideg} (\varepsilon_{rst})$. 
This shows that \eqref{c-direct-sum} is an orthogonal  
direct sum with respect to the sesquilinear form 
\eqref{define-sequi-form}, even though this property 
was not being considered when I defined 
\eqref{define-sequi-form}.  
However, this same analysis shows that 
the Hamel basis $\{ \varepsilon_{klm} \}$ is not 
an orthogonal basis, since for given indices 
$ k, l, m $ we have 
$
\langle 
\varepsilon_{klm}, \varepsilon_{kst} 
\rangle_{\mathcal{C}} 
\ne 0
$
for all pairs $ s,t \in \mathbb{N} $
satisfying $ s - t = l - m $. 
And there are infinitely many such pairs.  
It is known that there are other natural sesquilinear 
forms on $ \mathbb{C} $ for which 
\eqref{c-direct-sum} is an orthogonal direct sum. 
In fact, this is done using the (unique!) 
\textit{Haar state} 
of $ SU_q(2) $ and so is more closely related to 
the structure of  $ SU_q(2) $ as a quantum group. 
Again, see \cite{KS} for more details. 

We define $\mathcal{P}:= \mathrm{alg} \{a , c \}$, 
the sub-algebra (but not sub-$*$-algebra) 
of $SU_q (2)$ generated by $a$ and $c$. 
This is a sub-algebra of `holomorphic' elements.  
This is the same sub-algebra that was used in 
\cite{sbs5} for a Toeplitz quantization of $SU_q (2)$. 
We can identify $\mathcal{P}$ as the free algebra generated 
by $a$ and $c$, modulo the relation $a c = q \, c a$, and so 
(as an algebra) $\mathcal{P}$ is the 
complex Manin quantum plane, which  is denoted by 
$ A_{q}^{2|0} $ in \cite{manin}.

A Hamel basis of $\mathcal{P}$ 
is given by the monomials 
$a^k c^l = \varepsilon_{kl0}$ for $k,l \in \mathbb{N}$.
Since 
$$
\langle 
\varepsilon_{kl0} ,  \varepsilon_{rs0} 
\rangle_{\mathcal{C}} 
= w(k, l) \, \delta_{k,r} \, \delta_{l, s}, 
$$ 
we have that $\{ a^k c^l ~|~ k,l \in \mathbb{N} \}$ 
is an orthogonal basis of $\mathcal{P}$ 
and that the sesquilinear form 
$\langle \cdot , \cdot \rangle_{\mathcal{C}}$ 
when restricted to 
$\mathcal{P}$ is a positive definite inner product. 
Clearly, 
$$
    \phi_{kl} := \dfrac{1}{ \sqrt{w(k,l)} } \, a^k \, c^l 
    = \dfrac{1}{ \sqrt{w(k,l)} } \,  \varepsilon_{kl0} 
    \qquad \mathrm{for~} k,l \ge 0
$$
is an orthonormal basis of $\mathcal{P}$. 
Thus $\mathcal{P}$ is a pre-Hilbert space 
whose completion is denoted by $\mathcal{H}$. 
With no loss of generality we can assume 
that $\mathcal{P}$ is 
a dense subspace of $\mathcal{H}$.

The injection 
$j : \mathcal{P} \to \mathcal{C} $ is defined 
to be the inclusion map. 
The quotient map $Q: \mathcal{C} \to \mathcal{P} $ 
is defined as in \eqref{specific-P} by 
$$
       Q(f) := \sum_{i,j \ge 0} 
       \langle \phi_{ij}, f \rangle_{\mathcal{C}} \, \phi_{ij} 
$$
for $f \in \mathcal{C}$. 
The sum on the right side has only finitely 
many non-zero terms. 
It is now any easy exercise to prove that 
$Q(a) = a$, $Q(c) = c$ 
and $Q(a^*) = Q(c^*) = 0$, 
these being results needed to prove 
some of the statements in the next paragraph. 
We will discuss the action of $Q$ on the basis 
elements $\varepsilon_{klm}$ a little later on. 

According to the general theory of 
Section~\ref{co-toeplitz-quantization-section}, 
the projection $Q$ should be a co-algebra 
morphism, meaning a linear map intertwining 
the two co-multiplications. 
While $Q$ is clearly linear, 
we have not specified a co-multiplication 
$\Delta_\mathcal{P}$ on the Manin 
quantum plane $\mathcal{P}$. 
To do this we require that $\Delta_\mathcal{P}$ 
is the unique algebra morphism 
$\mathcal{P} \to \mathcal{P} \otimes \mathcal{P}$ 
satisfying 
\begin{equation*}
     \Delta_\mathcal{P} (a) := a \otimes a 
     \quad \mathrm{and} \quad 
     \Delta_\mathcal{P} (c) := c \otimes a. 
\end{equation*}
To see that this does make sense, 
one first defines the algebra morphism 
$\Delta_\mathcal{P}$ on the free algebra 
generated by $a$ and $c$ by using the 
previous formulas, 
and then one shows that 
$\Delta_\mathcal{P} (a c - q \, c a) = 0$. 
Hence $\Delta_\mathcal{P}$ passes to the 
quotient algebra $\mathcal{P}$. 
It is straightforward to show that 
$\Delta_\mathcal{P}$ so defined is co-associative. 
However, no linear map 
$ l : \mathcal{P} \to \mathbb{C} $
can be the co-unit for this co-multiplication, 
since 
$$
   (l \otimes id) \Delta_\mathcal{P} (c) =
   (l \otimes id) (c \otimes a) = l(c) a \ne c. 
$$
So, $ \mathcal{P} $ is a co-algebra without co-unit, 
which is allowed in the general theory. 
Finally, one can readily prove that 
$Q: \mathcal{C} \to \mathcal{P} $ 
is a co-algebra morphism and that 
$ \mathcal{P} $ is not a sub-co-algebra 
of $ \mathcal{C} $. 

We now calculate the action of $Q$ 
on the basis elements $\varepsilon_{klm}$ of 
the co-algebra $\mathcal{C} = SU_q(2)$: 
\begin{align*}
Q(\varepsilon_{klm}) &= \sum_{i,j \ge 0} 
\langle 
\phi_{ij}, \varepsilon_{klm} 
\rangle_{\mathcal{C}} \, \phi_{ij} 
= \sum_{i,j \ge 0} \dfrac{1}{w(i,j)} 
\langle 
\varepsilon_{ij0}, \varepsilon_{klm} 
\rangle_{\mathcal{C}} \, \varepsilon_{ij0} 
\\
&= \sum_{i,j \ge 0} \dfrac{1}{w(i,j)} w(i,j) \delta_{i,k} \delta_{j,l-m} \, \varepsilon_{ij0} 
= \sum_{i,j \ge 0} \delta_{i,k} \delta_{j,l-m} \, \varepsilon_{ij0}
\\
&= \varepsilon_{k,l-m, 0}. 
\end{align*}
We establish the convention from now on that 
$ \varepsilon_{rst} =0 $ if either $ r < 0 $
or $ s < 0 $. 
So the last result says $ Q(\varepsilon_{klm}) = 0 $ 
if $k < 0$ or $ l < m $. 

Summarizing.
we have shown the following: 

\begin{prop}
\label{pro1}
The action of the projection $ Q $ on the 
basis elements $ \varepsilon_{klm} $ is given by 
\begin{align*}
   Q(\varepsilon_{klm}) &= 
   \varepsilon_{k,l-m,0} \ne 0 \qquad \mathrm{if~} 
   k \ge 0, \, l \ge m,
   \\
   Q(\varepsilon_{klm}) &= 0 
   \hskip 2.65cm \mathrm{otherwise.}
\end{align*}
\end{prop}

In the case $k \ge 0$, one can interpret these  
formulas for $ Q(\varepsilon_{klm})$ 
as saying that all the $c^*$'s  
disappear and each one of them also 
`kills off' exactly one of the $ c $'s. 
The condition $l < m$ means that the monomial
$\varepsilon_{klm}$ has strictly more occurrences of 
$c^*$'s than of $c$'s, in which case 
all the $ c $'s get `killed off', as does everything else,  
and the result is $ 0 $. 
Finally, if $ k < 0 $, then there are occurrences of 
$ a^{*} $ but none of $ a $, 
and this in itself 
suffices to give $ 0 $. 
This last fact has a handy generalization, 
which we now present. 
\begin{prop}
\label{handy} 
Let $w$ be a finite word in the 
alphabet with these four letters: $a, a^*, c , c^*$. 
If $w$ has strictly more occurrences of 
the letter $a^*$ than of the letter $ a $, then $Q(w) = 0$. 
\end{prop}

\noindent 
\textit{Remark:} The hypothesis implies that the 
number of occurrences of $ a^{*} $ is strictly
larger than zero. 

\begin{proof}
Using the defining relations \eqref{SU-q-2-relations} 
we can push all occurrences of $c$ and $c^*$ to 
the right, thereby getting $w = q^{n} \, w^\prime \, c^l \, (c^*)^m$, where 
$l,m \in \mathbb{N}$, $n \in \mathbb{Z}$ and  
$w^\prime$ is a word with only occurrences of $a, a^*$. 
The number of occurrences of $a$ (resp., $a^*$) in $w^\prime$
is equal to the number of occurrences of $a$ (resp., $a^*$) in $w$. 
Let $ j $ be the number of occurrences of $ a^{*} $. 
We proceed by using induction on 
$k$, the number of occurrences of $a$ in $w$. 

First, we consider the case $k=0$. 
Then we have $w = q^{n} \, \varepsilon_{-j,l,m}$, 
where $j \ge k +1 = 1$ is 
the number of occurrences of $a^{*}$ in $w$. 
So, $Q(w) = 0$ by Proposition~\ref{pro1}.  

For the induction step we assume that the 
assertion $Q(w) = 0$ is true for some $k \ge 0$,
and then we will prove it for $k+1$. 
So, let $w$ be a word with $k+1 \ge 1$ 
occurrences of $a$.  
Then by hypothesis $  j > k + 1   $. 
We again have $w = q^{n} \, w^\prime \, c^l \, (c^*)^m$ as above. 
Since $w^\prime$ has a non-zero number 
of occurrences of both $a$ and $a^*$, 
we can write $w^\prime$ in at least one of these 
two forms:
$$
           w^\prime = u \, (a a^*) \, v  \quad \mathrm{or} \quad
           w^\prime = u \, (a^* a) \, v,
$$
where $u$ and $v$ are words (possibly empty)
with occurrences of $a$ and $a^*$ only. 
In the first case we see for example that
\begin{align*}
   &Q( w^\prime \, c^l \, (c^*)^m) = Q( u \, (a a^*) \, v  \, c^l \, (c^*)^m) 
   = Q( u \, (1 - q^2 c c^*) \, v  \, c^l \, (c^*)^m)
   \\
   &= Q( u \, v  \, c^l \, (c^*)^m) - q^2 Q( u \, (c c^*) \, v  \, c^l \, (c^*)^m)
   \\
   &= Q( u \, v  \, c^l \, (c^*)^m) - q^r Q( u \, v  \, c^{l+1} \, (c^*)^{m+1})
   = 0 - 0 = 0.
\end{align*}
Here the exponent $r \in \mathbb{N}$ 
arises from pushing 
the factor $c c^*$ to the right through $v$. 
The next to the last equality follows 
from the induction hypothesis and the 
fact that the word $u \, v$ has $k$ occurrences of $a$ 
and $ j - 1 > k \ge 0 $ occurrences of $ a^{*} $. 

The proof for the second form of $w^\prime$ is quite similar and so is
left to the reader. 
And that finishes the proof. 
\end{proof}

This result can also we proved 
by evaluating the bi-degree of a word with more 
$ a^{*} $'s than $ a $'s and 
showing that it is not equal to the bi-degree 
of any $ \varepsilon_{rs0} $ with $ r,s \ge 0 $. 

We have a result similar to Proposition \ref{handy}
for $ c $ and $ c^{*} $. 
\begin{prop}
	\label{handier} 
	Let $w$ be a finite word in the 
	alphabet with these four letters: $a, a^*, c , c^*$. 
	If $w$ has strictly more occurrences of 
	the letter $c^*$ than of the letter $ c$, then $Q(w) = 0$. 
\end{prop}
\begin{proof}
Here is a proof using bi-degrees instead on 
a similar induction argument, which could also be made. 
Suppose that $ w $ has $ j, k, l, m  $ 
occurrences of $ a, a^{*}, c, c^{*} $ 
respectively.
Then, independent of the order of 
these occurrences, we have that  
\begin{align*}
\mathrm{bideg} (w) &= 
j (1,1) + k (-1,-1) + l (-1,1) + m (1,-1) 
\\
&= ( j - k -l + m, j - k + l - m  ),
\end{align*}
while 
$ \mathrm{bideg} (a^r c^s)  = (r - s, r + s)$. 
The difference of the two entries in $ \mathrm{bideg} (w) $ 
is $ -2l + 2m > 0 $, since by hypothesis $ m > l $. 
However, the corresponding difference for 
$ \mathrm{bideg} (a^r c^s) $ is $ -2 s \le 0 $.  
This implies that 
$ \mathrm{bideg} (w) \ne \mathrm{bideg} (a^r c^s) $ and 
therefore $ \langle a^r c^s , w \rangle_{\mathcal{C}} = 0 $ 
for all $ r, s \ge 0 $, which in turn implies that 
$ Q(w) = 0 $. 
\end{proof} 

We have now on hand enough formulas to calculate the 
action of the co-Toeplitz operators 
$C_{\varepsilon_{klm}}$. 
This is sufficient information, since 
$C_g$ for 
any symbol $g \in SU_q(2)$ 
can be written as a 
finite linear combination with 
complex coefficients of 
the co-Toeplitz operators $C_{\varepsilon_{klm}}$. 
Moreover, it suffices to calculate 
$C_{\varepsilon_{klm}}$ acting on the 
elements $\phi_{r,s}$ in the standard 
orthonormal basis, where $r,s \in \mathbb{N}$. 
We recall that the co-Toeplitz operator 
with symbol $g$ was defined as 
$C_g = \pi_g \, \beta \, j$. 
Since $j$ is simply the inclusion map, we have 
$$
     C_{\varepsilon_{klm}} (\phi_{r,s}) = \pi_{\varepsilon_{klm}} \, \beta (\phi_{r,s}). 
$$
We will take the co-action map 
$\beta : \mathcal{C} \to 
\mathcal{P} \otimes \mathcal{C}$ 
to be of the form~\eqref{beta-special-form}, namely 
$$
       \mathcal{C} \stackrel{\Delta_\mathcal{C}}{\longrightarrow} \mathcal{C} \otimes \mathcal{C} 
       \stackrel{Q \otimes id}{\longrightarrow} \mathcal{P} \otimes \mathcal{C}, 
$$
where $\Delta_\mathcal{C}$ is the 
co-multiplication of $\mathcal{C}$. 
Dropping the normalization constant 
for the moment, we calculate with 
the monomial $a^r c^s$ instead of with $\phi_{r,s}$. 
We then see that 
\begin{align*} 
    \beta (a^r c^s) &= (Q \otimes id ) \big( \Delta_\mathcal{C} (a^r c^s) \big) 
    =  (Q \otimes id ) \big( \Delta_\mathcal{C} (a)^r  \Delta_\mathcal{C} (c)^s \big) 
    \\
    &=  (Q \otimes id ) \Big( (a \otimes a - q \, c^* \otimes c)^r (c \otimes a + a^* \otimes c)^s \Big) 
\end{align*}
We will use the standard binomial theorem on 
the second factor, since 
$ c \otimes a $ and $ a^{*} \otimes c $ commute, 
as follows from \eqref{SU-q-2-relations}. 
To continue with the first factor 
we will use the $q$-binomial theorem 
(see \cite{KS}), which states that 
if variables $ v, w $ satisfy the 
commutation relation $ v w = q w v  $ 
for $ 0 \ne q \in \mathbb{C} $, then 
for any integer $ n \ge 0 $ one has 
$$
     (v + w)^{n} = \sum_{m=0}^{n}
     \Big[ 
     \begin{array}{c}
     n \\ m
     \end{array}
     \Big]_{q^{-1}} v^{m}  w^{n-m}, 
$$
where the coefficient is an explicitly 
given deformation of the standard 
binomial coefficient. 
This is applicable in this situation, since
$$
    ( -q c^* \otimes c) (a \otimes a) 
    = -q c^* a \otimes c a,
$$
and hence for $ v = a \otimes a $ and 
$ w =  -q c^* \otimes c $ 
by using the relations  
\eqref{SU-q-2-relations} again 
we obtain 
\begin{align*}
     v w &= 
    (a \otimes a)( -q c^* \otimes c) 
    = -q a c^* \otimes a c 
    \\
    &= q^2 \big(  -q c^* a \otimes c a \big)
    \\
    &= q^{2} w v. 
\end{align*}

Next, to simplify somewhat 
the rather cumbersome binomial-type notation, 
we introduce 
$
B_{n,q} := \Big[ 
\begin{array}{c}
r \\ n
\end{array}
\Big]_{q^{-2}},  
$
which also suppresses the variable $ r $. 
We also use 
$
B_{p,1} := \Big( 
\begin{array}{c}
s \\ p
\end{array}
\Big),  
$
a standard binomial coefficient (which suppresses 
the variable $ s $). 
We will use this material in the next 
and subsequent calculations. 
The reader can consult \cite{KS} for more details 
about this so-called {\em $ q $-calculus}. 

Then for $r,s \in \mathbb{N}$ we have 
\begin{align*}
       &\beta (a^r c^s) = 
       (Q \otimes id ) \Big( (a \otimes a - q \, c^* \otimes c)^r (c \otimes a + a^{*} \otimes c)^s \Big)
       \\
       &= (Q \otimes id )  
       \sum_{n,p=0}^{r,s}
       \!\! B_{n,q} (a \otimes a)^{r-n} \, 
       (-q)^n (c^* \otimes c)^n 
       B_{p,1}
       (c \otimes a)^{s-p} (a^{*}\otimes c)^p 
       \\
       &= (Q \otimes id ) 
       \Big(  \sum_{n, p=0}^{r,s} 
       (-q)^n B_{n,q} B_{p,1} \, 
       a^{r-n} (c^*)^n c^{s-p}  a^{*p}
       \otimes a^{r-n} c^n a^{s-p} c^p \Big)
       \\
       &= 
      \sum_{n, p=0}^{r,s} 
       (-q)^n B_{n,q} B_{p,1} \, 
       Q ( a^{r-n} (c^*)^n c^{s-p}  a^{*p} )
       \otimes a^{r-n} c^n a^{s-p} c^p 
       \\
       &= 
       \sum_{n, p=0}^{r,s} \phi 
       \otimes a^{r-n} c^n a^{s-p} c^p.   
\end{align*}

To simplify notation we have put 
\begin{equation}
\label{define-phi-nprs}
\phi = \phi_{nprs} = (-q)^n B_{n,q} B_{p,1} \, 
Q ( a^{r-n} (c^*)^n c^{s-p}  a^{*p} ) \in \mathcal{P}. 
\end{equation}
By Propositions \ref{handy} and \ref{handier} we see that 
if $ p > r - n $ or $ n > s - p $, then $ \phi = 0 $. 
In the contrary case the calculation 
of $ \phi $ is a bit more complicated. 
The contrary case occurs when $ p \le r -n $ and 
$ n \le s -p $, that is, 
$ n + p \le r $ and $ n+p \le s $. 
This condition is then equivalent to 
$ n + p \le \min (r,s) $, which we will 
assume to hold throughout the following. 
The summation indices $ n $ and $ p $ also satisfy 
$ 0 \le n \le r $ and $ 0 \le p \le s $. 
To do this calculation we will use the identity 
\begin{equation}
\label{am-astarm-identity}
    a^{m} (a^{*})^{m} = \sum_{i=0}^{m} 
    \Big[  
    \begin{array}{c}
    m \\ i
    \end{array}
    \Big]_{q^{-2}}
    (-1)^{i} q^{i + 2 i m -i^2} c^{i} (c^{*})^i, 
\end{equation}
for integer $ m \ge 0 $. 
(Cp. \cite{KS}, p.~100, Eq.~ (13).
Or prove it yourself by induction on $ m $.)
Note that this identity is not surprising, 
since $ \mathrm{bideg} (a^{m} (a^{*})^{m}) = (0,0)$ 
and $ A[0,0] $ is the polynomial algebra in the 
variable $ c c^{*} $. 
(Recall that $ c $ and $ c^{*} $ commute 
so that $ c^{i} (c^{*})^i = ( c c^{*})^{i} $.) 
What the identity \eqref{am-astarm-identity} 
tells us more specifically 
is that  $ a^{m} (a^{*})^{m} $ is a polynomial of
degree $ m $ and what its coefficients are exactly.  
Then using this identity we have 
\begin{align*}
   &a^{r-n} (c^*)^n c^{s-p}  a^{*p} = 
   q^{p(s-p) + p n} a^{r-n}  a^{*p}  (c^*)^n c^{s-p}  
   \\
   &=  
   q^{p(s-p + n)} a^{r-n-p} a^{p} a^{*p}
    (c^*)^n c^{s-p}   
   \\
   &=  
   q^{p(s-p + n)} a^{r-n-p} 
   \sum_{i=0}^{p} 
   \Big[  
   \begin{array}{c}
   p \\ i
   \end{array}
   \Big]_{q^{-2}}
   (-1)^{i} q^{i + 2 i p -i^2} c^{i} (c^{*})^i 
   (c^*)^n c^{s-p} 
   \\
   &=
   \sum_{i=0}^{p} 
   \Big[  
   \begin{array}{c}
   p \\ i 
   \end{array}
   \Big]_{q^{-2}}
   (-1)^{i} q^{A} a^{r-n-p} 
   c^{i + s -p} (c^{*})^{i + n}
   \\
   &=
   \sum_{i=0}^{p} 
   \Big[  
   \begin{array}{c}
   p \\ i
   \end{array}
   \Big]_{q^{-2}}
   (-1)^{i} q^{A} 
   \varepsilon_{r-n-p, i + s -p, i + n}, 
\end{align*}
where $ A = p(s-p + n) + i + 2 i p - i^2 $. 
Continuing, we see that 
\begin{align*}
\phi_{nprs} &= (-q)^n B_{n,q} B_{p,1} \, 
Q ( a^{r-n} (c^*)^n c^{s-p}  a^{*p} )
\\
&= (-q)^n B_{n,q} B_{p,1} \, \left( 
\sum_{i=0}^{p} 
\Big[  
\begin{array}{c}
p \\ i 
\end{array}
\Big]_{q^{-2}}
(-1)^{i} q^{A} 
Q ( \varepsilon_{r-n-p, i +s -p , i + n} ) \right) 
\\
&= (-q)^n B_{n,q} B_{p,1} \, \left( 
\sum_{i=0}^{p} 
\Big[  
\begin{array}{c}
p \\ i 
\end{array}
\Big]_{q^{-2}}
(-1)^{i} q^{A} \right) 
\varepsilon_{r-n-p, s -n -p , 0}
\\
&= D_{nprs} \, \varepsilon_{r - (n+p), s - (n+p) , 0}, 
\end{align*}
where the real number $ D_{nprs} $ has the 
obvious definition. 
Here we also used Proposition \ref{pro1}, which 
has the fortuitous virtue of changing the scope 
of the sum on $ i $. 
Notice that this shows that 
$ \phi $ is proportional to an element in the 
basis $\{ \varepsilon_{kl0} ~|~ k,l \ge 0 \}$
of $ \mathcal{P} $. 
The bi-degree 
of the bi-homogeneous element 
$ \phi $ is easily seen to be given by 
\begin{equation}
\mathrm{bideg} (\phi) =  
\mathrm{bideg} (\varepsilon_{r- (n+p) , s -(n+p) , 0}) 
= ( \, r - s   , r + s - 2 (n+p) \, ). 
\label{bideg-phi}
\end{equation}

Next, for $r,s \in \mathbb{N}$ we obtain 
\begin{align}
     &C_{\varepsilon_{klm}} (a^r c^s) = \pi_{\varepsilon_{klm}} \, \beta (a^r c^s) \nonumber  
     \\
     &= \pi_{\varepsilon_{klm}} \sum_{n+p=0}^{\min(r,s)} 
     \phi \otimes a^{r-n} c^n a^{s-p} c^p \nonumber 
     \\
     &= \sum_{n+p=0}^{\min(r,s)} 
     \langle \varepsilon_{klm} , a^{r-n} c^n a^{s-p} c^p \rangle_{\mathcal{C}} \, \phi \nonumber 
     \\
     &= \sum_{n+p=0}^{\min(r,s)} q^{n (s-p)}
     \langle \varepsilon_{klm} , a^{r+s - (n+p)} c^{n+p} \rangle_{\mathcal{C}} \, \phi_{nprs} 
    \nonumber 
\\
 &= \sum_{n+p=0}^{\min(r,s)} q^{n (s-p)}
 \langle \varepsilon_{klm} , a^{r+s - (n+p)} c^{n+p} \rangle_{\mathcal{C}} D_{nprs}
\varepsilon_{r- (n+p) , s - (n+p) , 0}.
 \label{last-expn} 
\end{align}

Note that the condition $ 0 \le n+p \le \min (r,s) $
means according to \eqref{bideg-phi} that 
\eqref{last-expn} is in general 
a sum of bi-homogeneous 
elements with different bi-degrees. 
However, the coefficients of these summands 
will be non-zero only if 
the inner product in the expression \eqref{last-expn}  
is non-zero which is equivalent to 
$$
   \mathrm{bideg} (\varepsilon_{klm}) = 
   \mathrm{bideg} ( a^{r+s - (n+p)} c^{n+p} ), 
$$
which itself is equivalent to
$$
(k - l + m, k + l  - m ) =
( \, r + s - 2 (n+p) , r + s \, ). 
$$
The indices $ k,l,m, r, s $ are given and the 
`unknowns' are the summation indices $ n $ 
and $ p $. 
The previous equality is equivalent to 
\begin{equation}
\label{proportional-term}
   n+p = r + s - k =  l - m. 
\end{equation}
If this holds for some pair $ n,p $
satisfying $ n + p \le  \min(r,s)$,  
$ 0 \le n \le r $ and $ 0 \le p \le s $, 
then \eqref{last-expn} is a multiple of 
$$
\varepsilon_{r- (n+p) , s - (n+p) , 0} = 
\varepsilon_{r- (l - m) , s - (l - m) , 0}; 
$$
otherwise, \eqref{last-expn} is $ 0 $. 
In order that there exists at least one solution 
of \eqref{proportional-term} 
for a pair $ n \ge 0, \, p \ge 0 $
it is necessary and sufficient that the five
indices  $ k,l,m, r, s $ satisfy 
\begin{equation}
\label{klmrs-condition}
k \le r + s \quad \mathrm{and} \quad m \le l.  
\end{equation}

And in that case the co-Toeplitz operator 
$ C_{\varepsilon_{klm}} $ lowers the 
degree of each variable $ a,c $ by $ l-m \ge 0 $. 
Alternatively, we note that 
$ C_{\varepsilon_{klm}} $ maps
$ a^r c^s = \varepsilon_{rs0} $ of bi-degree $ (r-s, r+s) $ to 
$ \varepsilon_{r- (l - m) , s - (l - m) , 0} $, an 
element of bi-degree $ (r-s, r + s - 2 (l-m ) )$. 
In other words on this scale the co-Toeplitz operator 
$ C_{\varepsilon_{klm}} $ can be understood as 
an operator having bi-degree $ (0, - 2 (l-m )) $. 
In physics terminology, these co-Toeplitz operators
are not {\em creation operators}
in the sense that the degree of the powers of 
monomials is strictly increased. 
Similarly, the bi-degree also is not strictly 
increased. 

We have shown the following.
\begin{theorem}
Suppose $ k \in \mathbb{Z} $ and $ l,m,r,s \in \mathbb{N} $ satisfy $ r + s - k = l - m  $ and 
$ 0 \le l - m \le \min (r,s) $. 
Suppose that this set is non-empty:
$$
\{ (n,p) ~|~ n + p = l -m, \, 0 \le n \le r, \, 
  0 \le p \le s  \}
$$
Then 
$ C_{\varepsilon_{klm}} (a^r c^s) = 
K a^{r-(l-m)} c^{r-(l-m)} $
for some real number $ K $. 

In terms of basis elements
$ C_{\varepsilon_{klm}} (\phi_{rs} ) = K^{\prime}
\phi_{r-(l-m), s-(l-m)}$
for some real number $ K^{\prime} $. 
And $ K^{\prime} \ne 0$ if and only if $ K \ne 0 $. 

Otherwise, we have 
$ C_{\varepsilon_{klm}} (a^r c^s) = 0 $. 
\end{theorem}

Here are some special cases of this theorem. 
First, we consider the case $ l = m $. 
In this case $ C_{\varepsilon_{kll}} $
maps $ a^r c^s $ to a multiple of $ a^r c^s $ 
for any value of $ k \in \mathbb{Z} $.  
Notice that the multiplicative constant 
depends on $ k $ and can be $ 0 $.
In physics terminology this is a 
{\em preservation operator}, which 
simply means mathematically 
that it preserves degrees. 
The sub-case $ l = m = 0 $ is the co-Toeplitz operator 
with `holomorphic' symbol $ a^{k} $ 
if $ k \ge 0 $ or 
with `anti-holomorphic' 
symbol $ (a^{*})^{-k} $ if $ k < 0 $. 

The next case is $ l > 0,  m = 0 $. 
In this case $ C_{\varepsilon_{kl0}} $
maps $ a^r c^s $ to some 
multiple of $ a^{r-l} c^{s-l} $. 
In usual physics terminology this is called an 
{\em annihilation operator}, which 
simply means mathematically 
that it lowers degrees.  
We remark that $ \varepsilon_{kl0} $ is the most
general holomorphic monomial in the variables
$ a $ and $ c $.  
It is because of this particular case that we have 
defined co-Toeplitz operators 
with holomorphic symbols to be annihilation 
operators. 

In the case $ l = 0 $ we have that $ m = 0 $
must hold as well. 
And so this case was already considered as 
part of the first case. 
Or in other words, the case $ l = 0 $ and $ m > 0 $ 
gives a zero co-Toeplitz operator.

This leads us up to the analysis of the co-Toeplitz 
operators whose symbols are one of the four
algebra generators, $ a, a^{*}, c , c^{*} $, of
$ SU_q(2) $. 
For the symbol $ c $ we have $ k = m = 0 $, 
$ l=1 $ and so
$ C_c $ is an annihilation operator that maps
$ a^r c^s $ to a multiple of $ a^{r-1} c^{s-1} $. 

For the symbol $ c^{*} $ we have $ k = l = 0 $, 
$ m = 1 $ and so $ C_{c^{*}} = 0$, since $ m > l $ 
holds. 
The same reasoning applies to the 
`anti-holomorphic' symbol 
$ (a^{*})^{k} (c^{*})^{m} $ for $ m >0 $, 
since $ m > l = 0 $. 
So, $ C_{ (a^{*})^{k}(c^{*})^{m} } = 0$. 

For the symbol $ a^{*} $ we have 
$ l = m = 0 $, $ k = -1 $. 
Now $ n + p = l - m = 0 $ implies that 
$ n = p = 0 $ and therefore that 
$ r + s = k = -1 $, which 
has no solutions
$ r \ge 0, s \ge 0 $. 
Thus, $ C_{a^{*}} = 0 $. 

For the symbol being $ a $ 
we have $ l = m = 0 $, $ k = 1 $, 
and thus $ C_{a} $ is a  
preservation operator. 
But $ n + p = l - m = 0 $ implies that 
$ n = p = 0 $. 
So there is only one term in the sum \eqref{last-expn}. 
We note that $ q^{n(s-p)} = q^{0} = 1 $
and $ D_{00rs} = 1 $. 
But the coefficient in that unique term is
$$
\langle \varepsilon_{100} , a \rangle = 
\langle a , a \rangle = w (1,0) > 0. 
$$
Consequently, $ C_{a} $ is a non-zero multiple 
of the identity operator. 
In particular, 
$ C_{a} \ne 0$ and $ C_{a^{*}} = 0 $
are not adjoints of each other. 
So the condition 
\eqref{M-tilde-condition}
does not hold for our choice 
\eqref{define-sequi-form} 
for the sesquilinear form. 

In this example, the creation and annihilation 
operators have strange properties from the point 
of view of quantum physics. 
This is in part a consequence of the choice of the 
sesquilinear form for this example. 
As I have emphasized elsewhere, the study of 
more examples of the co-Toeplitz 
quantization scheme is really needed 
for getting a better 
understanding of the general theory. 
A similar example for the Toeplitz quantization  
of $ SU_q (2) $ in \cite{sbs5} gave creation and 
annihilation operators 
which are more intuitive physically.  
This goes to show that co-Toeplitz quantization 
has new, rather curious properties, even though it is 
dual in the sense of notion duality 
to Toeplitz quantization. 

This example depends on more than the choice of 
the co-algebra $ SU_q(2) $. 
We have to choose also the sesquilinear form and 
the subspace $ \mathcal{P} $. 
We could continue with the same family 
of sesquilinear forms, where that family is 
parameterized by the weight function. 
Instead, we could use a different subspace, say 
for example:
$$
    \mathcal{P}^{\prime} := 
    \mathrm{span}\{ \varepsilon_{kl0} = a^{k} c^{l}, 
     \varepsilon_{k0m} = a^{k} (c^{*})^{m} 
    ~|~ k,l \ge 0, m > 0 \}.
$$ 
Since no two elements in this set of generators 
have the same bi-degree, we have that this is an orthogonal 
basis of $ \mathcal{P}^{\prime} $. 
So an orthonormal basis of $ \mathcal{P}^{\prime} $ 
is given by 
$$
\phi_{kl} 
= \dfrac{1}{ \sqrt{w(k,l)} } \,  \varepsilon_{kl0} 
\quad \mathrm{and} \quad 
\psi_{km} := 
\dfrac{1}{\sqrt{w(k,-m)}} \, \varepsilon_{k0m}
$$
for $ k,l \ge 0 $ and $ m > 0 $, where we continue 
to use the notation $ \phi_{kl}  $
introduced earlier. 
Thus $\mathcal{P}^{\prime}$ is a pre-Hilbert space. 

The injection 
$j^{\prime} : \mathcal{P}^{\prime} \to \mathcal{C} $ 
is defined to be the inclusion map. 
The quotient map 
$Q^{\prime}: \mathcal{C} \to \mathcal{P}^{\prime} $ 
is defined for $f \in \mathcal{C}$ as 
$$
Q^{\prime}(f) := \sum_{i,j \ge 0} 
\langle \phi_{ij}, f \rangle_{\mathcal{C}} \, \phi_{ij} 
+
\sum_{i \ge 0, j > 0} 
\langle \psi_{ij}, f \rangle_{\mathcal{C}} \, \psi_{ij},
$$
where the sum on the right side has only finitely 
many non-zero terms. 
This shows just one possible way of giving 
another example based on the co-algebra 
$ SU_q(2) $. 

Another possible modification of this example 
is to use the positive definite
inner product 
defined for $ x,y \in SU_q(2) $ by 
$
       \langle x , y \rangle := h (x^{*}y), 
$
where $ h : SU_q(2) \to \mathbb{C} $ is the 
unique Haar state on $ SU_q(2) $ 
(see \cite{KS}), instead 
of the sesquilinear form defined in \eqref{define-sequi-form}. 
This is an approach that is better attuned to 
the Hopf $ * $-algebra structure of $ SU_q(2) $. 
These two alternatives as well as other examples 
of co-Toeplitz quantizations of co-algebras 
will be the subject of forthcoming 
research work.

\section{Concluding Remarks}
\label{concluding-remarks-section}

This paper begins the new theory of 
co-Toeplitz operators and their associated 
quantization, as the title indicates. 
On the other hand, the theory of Toeplitz 
operators is over one hundred years old. 
Obviously, one strategy is to use the ideas and 
results in the Toeplitz setting to inspire 
research in this new theory. 
However, I hope that there will be more new ideas 
arising in the co-Toeplitz setting and that some 
of these may even shed light on the well-known 
Toeplitz setting. 
To bring this theory to maturity requires more 
than anything a reasonable quantity of 
illuminating examples, which could help in 
fine tuning definitions and in providing insights 
into relations among the various structures 
introduced here. 
Also, bi-algebras can now be quantized either by 
using their algebra structure or their co-algebra  
structure. 
So it would be interesting to understand how 
those two quantizations might be related. 
In the more specific case of Hopf algebras 
(or quantum groups) 
one would like to know what the role of the 
antipode is. 
One might also be able to introduce into this 
setting such structures as a symplectic 
form, Poisson brackets or coherent states, just 
to name a few possibilities. 
Finally, other types of quantization schemes may 
also be extended to theories based 
on arbitrary algebras or co-algebras. 
This is a broad outline of possible future 
research in this area. 

\vskip 0.2cm 

\begin{center}
\textbf{Acknowledgments}
\end{center}
I thank Micho \dju~and Jean-Pierre Gazeau 
for providing me 
insights from rather 
complementary points of view 
of mathematical physics. 
I can not imagine how I could ever have possibly 
written this paper without their generosity 
in sharing ideas with me.

\end{document}